\documentclass[letterpaper]{article}
\usepackage{aaai19}\nocopyright
\usepackage{times}
\usepackage{helvet}
\usepackage{courier}
\usepackage{url}
\usepackage{graphicx}
\title{Strong Equivalence for Epistemic Logic Programs Made Easy (Extended
Version)}
\author{Wolfgang Faber \and Michael Morak\\[1ex]
Alpen-Adria-Universit\"at Klagenfurt\\
Klagenfurt, Austria
\And
Stefan Woltran\\[1ex]
TU Wien\\
Vienna, Austria
}
\usepackage{amsthm,amsmath,amssymb,wasysym}

\newcommand{\nop}[1]{}

\newenvironment{changemargin}[2]{%
\list{}{\rightmargin#2\leftmargin#1
\parsep=0pt\topsep=0pt\partopsep=0pt}
\item[]}
{\endlist}

\newenvironment{indented}{\begin{changemargin}{1cm}{0cm}}{\end{changemargin}}

\newtheorem{theorem}{Theorem}
\newtheorem{corollary}[theorem]{Corollary}
\newtheorem{proposition}[theorem]{Proposition}
\newtheorem{lemma}[theorem]{Lemma}
\newtheorem{definition}[theorem]{Definition}
\newtheorem{example}[theorem]{Example}

\let\phi\varphi
\let\epsilon\varepsilon

\renewcommand{\models}{\vDash}

\newcommand{\calA}{\mathcal{A}}

\newcommand{\calC}{\mathcal{C}}

\newcommand{\calE}{\mathcal{E}}

\newcommand{\calI}{\mathcal{I}}

\newcommand{\calM}{\mathcal{M}}

\newcommand{\calR}{\mathcal{R}}
\newcommand{\calS}{\mathcal{S}}

\newcommand{\calY}{\mathcal{Y}}

\newcommand{\NP}{\ensuremath{\textsc{NP}}}
\newcommand{\co}{\ensuremath{\textsc{co}}}

\newcommand{\SIGMA}[2]{\ensuremath{\Sigma_{\mathit{#1}}^{\mathit{#2}}}}

\newcommand{\mods}[1]{\mathit{mods}(#1)}
\newcommand{\answersets}[1]{\mathit{AS}(#1)}
\newcommand{\semods}[1]{\calS\calE(#1)}
\newcommand{\cwvs}[1]{\mathit{cwv}(#1)}

\newcommand{\eneg}{\mathbf{not}\,}

\newcommand{\sneg}{\mathbf{\sim}\,}

\newcommand{\body}[1]{{\mathit{B}(#1)}}

\newcommand{\pbody}[1]{{\mathit{B}^+(#1)}}
\newcommand{\head}[1]{{\mathit{H}(#1)}}

\begin{document}

\maketitle

\begin{abstract}
  Epistemic Logic Programs (ELPs), that is, Answer Set Programming (ASP)
extended with epistemic operators, have received renewed interest in recent
years, which led to a flurry of new research, as well as efficient solvers. An
important question is under which conditions a sub-program can be replaced by
another one without changing the meaning, in any context. This problem is
known as strong equivalence, and is well-studied for ASP. For ELPs, this
question has been approached by embedding them into epistemic extensions of
equilibrium logics. In this paper, we consider a simpler, more direct
characterization that is directly applicable to the language used in
state-of-the-art ELP solvers. This also allows us to give tight complexity
bounds, showing that 
%, surprisingly, 
strong equivalence for ELPs remains
coNP-complete, as for ASP. We further use our results to provide syntactic
characterizations for tautological rules and rule subsumption for ELPs.

\end{abstract}

\section{Introduction}\label{sec:introduction}

Epistemic Logic Programs (ELPs) are an extension of the
well-established formalism of Answer Set Programming (ASP), a generic,
fully declarative logic programming language that allows for encoding
problems such that the resulting answers (called \emph{answer sets})
directly correspond to solutions of the encoded problem
\cite{cacm:BrewkaET11,ki:SchaubW18}. Negation in ASP is generally interpreted
according to the stable model semantics \cite{iclp:GelfondL88}, that
is, as negation-as-failure, also called default negation. Such a
default negation $\neg a$ of an atom $a$ is true if there is no
justification for $a$ in the same answer set, making it a ``local''
operator in the sense that it is defined relative to one considered
answer set. ELPs (in the version of \cite{ai:ShenE16}), on the other
hand, extend ASP with the epistemic negation operator $\eneg$ that
allows for a form of meta-reasoning, that is, reasoning over multiple
answer sets. Intuitively, an epistemically negated atom $\eneg a$
expresses that $a$ cannot be \emph{proven} true, meaning that it is
not true in every answer set. Thus, epistemic negation is defined
relative to a collection of answer sets, referred to as a \emph{world
  view}. %The main reasoning task for ELPs, namely checking that a
Deciding whether a world view exists, is
\SIGMA{P}{3}-complete~\cite{ai:ShenE16}, one level higher on the
polynomial hierarchy than deciding answer set
existence~\cite{amai:EiterG95}.

\citeauthor{aaai:Gelfond91} \shortcite{aaai:Gelfond91,amai:Gelfond94}
recognized epistemic negation as a desired construct for ASP early on
and introduced the modal operators $\mathbf{K}$ (``known'' or
``provably true'') and $\mathbf{M}$ (``possible'' or ``not provably
false'') to address this. $\mathbf{K}a$ and $\mathbf{M}a$ correspond
to $\neg \eneg a$ and $\eneg \neg a$, respectively.  Renewed interest
in recent years has revealed several flaws in the original semantics,
and various new approaches
(cf.\ e.g.\ \cite{lpnmr:Gelfond11,birthday:Truszczynski11,diss:Kahl14,ijcai:CerroHS15,ai:ShenE16})
were proposed.
%aimed to refine them in such a way that unintended world views are
%eliminated.
%This flurry of new research also 
Also the development of ELP solving
systems \cite{logcom:KahlWBGZ15,ijcai:SonLKL17,ijcai:BichlerMW18} has gained momentum.
%Our work is based on the semantics proposed by \citeauthor{ai:ShenE16}
%\shortcite{ai:ShenE16}.

A main application and major motivation of ELPs is the formalization
of the Closed World Assumption (CWA), as pointed out already by
\citeauthor{aaai:Gelfond91}
\shortcite{aaai:Gelfond91}. Interestingly, there are two
formalizations of the CWA using ELPs in the literature. The first,
given in \cite{aaai:Gelfond91}%
\footnote{\citeauthor{aaai:Gelfond91}
\shortcite{aaai:Gelfond91} proposed the rule $\sneg p \gets \neg \mathbf{M} p$, where $\sneg$ is a third kind of negation, usually
referred to as strong negation, not considered in this paper. It can be
simulated by replacing occurrences of $\sneg p$ by a fresh atom
$p'$ and adding a constraint rule $\gets p, p'$ that excludes $p$ and
$p'$ to hold simultaneously.
%\citeauthor{aaai:Gelfond91} \shortcite{aaai:Gelfond91} also considered a first-order language and therefore introduced a rule for each predicate, rather than atom. These differences are not important for our argument, though.
}%
, shall be referred to as Gelfond-CWA and introduces one rule for each atom $p$:
\begin{equation}\label{eq:gelfond-cwa}
  p' \gets \neg \eneg \neg p.
\end{equation}
Intuitively, this says  that $p'$ (i.e.\ the negation of $p$) shall be true if
there is no possible world where $p$ is true.
%Interestingly,  
\citeauthor{ai:ShenE16} \shortcite{ai:ShenE16} propose a different rule for CWA, which we will refer to as Shen-Eiter-CWA for $p$:
\begin{equation}\label{eq:shen-eiter-cwa}
  p' \gets \eneg p.
\end{equation}
Again, intuitively, this formulation makes $p'$ true iff there is a possible
world where $p$ is false.
%Note that Shen-Eiter-CWA for $p$ corresponds to $\sneg p \gets \mathbf{M} \neg p$ in the language of \cite{aaai:Gelfond91}.

A natural question is whether $(\ref{eq:gelfond-cwa})$ and $(\ref{eq:shen-eiter-cwa})$ yield the same results in any context, which will be answered in this paper. 
%
%
%The main contribution of this paper is to 
To this end, we study %in this paper 
notions of equivalence between ELPs. For instance, two ELPs
$\Pi_1$ and $\Pi_2$ are strongly ELP-WV-equivalent iff, for any third ELP $\Pi$, the
combined programs $\Pi_1 \cup \Pi$ and $\Pi_2 \cup \Pi$ are equivalent (i.e.\
have the same world views). This notion is useful to transform ELPs into
equivalent versions where one wants to verify that a local change preserves
equivalence without considering the whole program. Other notions of strong equivalence can be defined for comparing candidate world views (rather than world views) or considering only the addition of programs that do not contain epistemic operators.
In (plain) ASP, strong equivalence is a well-studied problem.
%known as strong equivalence 
\citeauthor{tocl:LifschitzPV01} 
\shortcite{tocl:LifschitzPV01} 
have provided
an elegant  characterization of the problem in terms of Heyting's logic of here-and-there (HT).
Strong equivalence %not only shed light on the fundamentals %of ASP, but 
also proved useful as a means to simplify programs \cite{iclp:CabalarPV07,jair:LinC07,jancl:EiterFPTW13}.

It has been shown in
\cite{lpnmr:WangZ05} and \cite{ijcai:CerroHS15}, among others, that epistemic
extensions of %equilibrium logics 
logic HT
can be used to characterize strong equivalence
of ELPs. These approaches, however, are very general and lead to a very abstract
characterization that cannot be immediately used for ELPs written in the
language of current solving systems. It is also not easy to obtain tight
%computational 
complexity results in such a general setting. The semantics considered in \cite{lpnmr:WangZ05} is the original one of \cite{aaai:Gelfond91}, which is now %generally 
considered obsolete, while \cite{ijcai:CerroHS15} consider a different semantics
from the one in \cite{ai:ShenE16}, which is what we use. An in-depth comparison of the
differences in the semantics can also be found in \cite{ai:ShenE16}.

In this paper, we therefore propose a simpler, more direct
characterization for the well-understood ELP semantics given in
\cite{ai:ShenE16}. Our characterization is in the spirit of \cite{tplp:Turner03}, which
is useful to study ELPs written in the input language of the ELP
solvers mentioned above, as it can be directly applied in this
setting. This also allows us to obtain tight
complexity bounds for checking strong equivalence of ELPs. We further
investigate several use cases of strong equivalence by using our
technique to syntactically characterize tautological rules and rule
subsumptions.

\nop{
\paragraph{Related Work.} Strong equivalence for ASP and several extensions has
been studied for different formalisms: the landmark paper
\cite{tocl:LifschitzPV01} showed that the logic of here-and-there can be used to
characterize strong equivalence for logic programs; \cite{tplp:Turner03}
simplified these notions and gave a more direct characterization;
\cite{amai:EiterFFW07} proposes a characterization for logic programs with weak
constraints; \citeauthor{kr:FaberTW08} \shortcite{kr:FaberTW08} studies
different notion of equivalence for logic programs with ordered disjunction;
this was further generalized in \cite{jair:FaberTW13}.  Other research on strong
equivalence includes results for streaming programs \cite{ijcai:BeckDE16} and
temporal programs \cite{jelia:AguadoCPV08}.
}

\paragraph{Contributions.} The main contributions of this paper are the
following:
\begin{itemize}
  \item We propose different notions of strong equivalence of ELPs 
 	(based on the input language of %state-of-the-art 
        ELP solvers)
	that strictly generalize strong equivalence for plain ASP.

  \item We provide a model-theoretic characterization of strong equivalence for ELPs
	showing that the different notions proposed coincide.
  \item We use our characterization to show that, surprisingly, testing strong
    equivalence of two ELPs remains in \co\NP, that is, the complexity of this
    test does not increase when considering ELPs instead of plain ASP.
  \item Finally, we use our proposed notion to syntactically characterize
    tautological ELP rules and when one ELP rule subsumes another.
\end{itemize}

\nop{
\paragraph{Organization.} The remainder of the paper is structured as follows.
In Section~\ref{sec:preliminaries}, we give an overview of ASP, ELPs, and strong
equivalence for the former. In Section~\ref{sec:properties}, several useful
properties of ELPs are established that are then used in
Section~\ref{sec:strongequivalence} to propose our characterization for strong
equivalence of ELPs. Section~\ref{sec:complexity} then deals with the
computational complexity of strong equivalence checking. Finally, we investigate
several use cases, namely rule tautologies and rule subsumption, in
Section~\ref{sec:casestudies}, and close with some concluding remarks in
Section~\ref{sec:conclusions}.
}

\section{Preliminaries}\label{sec:preliminaries}

\paragraph{Answer Set Programming (ASP).} A \emph{ground logic program} with
nested negation (also called answer set program, ASP program, or, simply, logic
program) is a pair $\Pi = (\calA, \calR)$, where $\calA$ is a set of
propositional (i.e.\ ground) atoms and $\calR$ is a set of rules of the form
\begin{equation}\label{eq:rule}
  a_1\vee \cdots \vee a_l \leftarrow a_{l+1}, \ldots, a_m, \neg \ell_1, \ldots,
  \neg \ell_n;
\end{equation}
where the comma symbol stands for conjunction, $0 \leq l \leq m$, $0 \leq n$,
$a_i \in \calA$ for all $1 \leq i \leq m$, and each $\ell_i$ is a
\emph{literal}, that is, either an atom $a$ or its (default) negation $\neg a$
for any atom $a \in \calA$. Note that, therefore, doubly negated atoms
may occur. We will sometimes refer to such rules as \emph{standard rules}.
Each rule $r \in \calR$ of form~(\ref{eq:rule}) consists of a \emph{head}
$\head{r} = \{ a_1,\ldots,a_l \}$ and a \emph{body} $\body{r} =
\{a_{l+1},\ldots,a_m, \neg \ell_1, \ldots, \neg \ell_n \}$. We denote the
\emph{positive} body by $\pbody{r} = \{ a_{l+1}, \ldots, a_m \}$.

An \emph{interpretation} $I$ (over $\calA$) is a set of atoms, that is, $I \subseteq \calA$.  A
literal $\ell$ is true in an interpretation $I \subseteq \calA$, denoted $I
\models \ell$, if $a \in I$ and $\ell = a$, or if $a \not\in I$ and $\ell = \neg
a$; otherwise $\ell$ is false in $I$, denoted $I \not\models \ell$. Finally, for
some literal $\ell$, we define that $I \models \neg \ell$ if $I \not\models
\ell$. This notation naturally extends to sets of literals. An interpretation
$M$ is called a \emph{model} of $r$, denoted $M \models r$, if, whenever $M
\models \body{r}$, it holds that $M \models \head{r}$. We denote the set of
models of $r$ by $\mods{r}$; the models of a logic program $\Pi= (\calA,\calR)$
are given by $\mods{\Pi} = \bigcap_{r \in \calR} \mods{r}$. We also write
$I\models r$ (resp.\ $I\models \Pi$) if $I\in\mods{r}$ (resp.\
$I\in\mods{\Pi}$).

The GL-reduct $\Pi^I$ of a %ground 
logic program $\Pi = (\calA, \calR)$ with
respect to an interpretation $I$ is the program $%\Pi^I = 
(\calA, \calR^I)$,
where $\calR^I = \{ \head{r} \leftarrow \pbody{r} \mid r \in \calR, \forall \neg
\ell \in \body{r} : I \models \neg \ell \}$.

\begin{definition}\label{def:answerset}
  \cite{iclp:GelfondL88,ngc:GelfondL91,amai:LifschitzTT99} $M \subseteq \calA$
  is an \emph{answer set} of a logic program $\Pi$ if (1) $M \in \mods{\Pi}$ and
  (2) there is no subset $M' \subset M$ such that $M' \in \mods{\Pi^M}$.
\end{definition}

The set of answer sets of a logic program $\Pi$ is denoted by $\answersets{\Pi}$.
The \emph{consistency problem} of ASP, that is, to decide whether for a given
logic program $\Pi$ it holds that $\answersets{\Pi}\neq\emptyset$, is
\SIGMA{P}{2}-complete~\cite{amai:EiterG95}, and remains so also in the case
where doubly negated atoms are allowed in rule bodies~\cite{tplp:PearceTW09}.

\paragraph{Strong Equivalence for Logic Programs.} Two logic programs $\Pi_1 =
(\calA, \calR_1)$ and $\Pi_2 = (\calA, \calR_2)$ are \emph{equivalent} iff they
have the same set of answer sets, that is, $\answersets{\Pi_1} =
\answersets{\Pi_2}$. The two logic programs are \emph{strongly equivalent} iff
for any third logic program $\Pi = (\calA, \calR)$ it holds that the combined
logic program $\Pi_1 \cup \Pi = (\calA, \calR_1 \cup \calR)$ is equivalent to
the combined logic program $\Pi_2 \cup \Pi = (\calA, \calR_2 \cup \calR)$. Note
that strong equivalence implies equivalence, since the empty program
$\Pi=(\calA,\emptyset)$ would already contradict strong equivalence for two
non-equivalent programs $\Pi_1$ and $\Pi_2$.

An \emph{SE-model} \cite{tplp:Turner03} of a logic program $\Pi = (\calA, \calR)$ 
is a tuple %of interpretations 
$(X, Y)$, where 
$X \subseteq Y\subseteq\calA$, 
$Y \models \Pi$, 
and $X \models \Pi^Y$. The set of SE-models of a logic program
$\Pi$ is denoted $\semods{\Pi}$. Note that for every model $Y$ of $\Pi$, $(Y, Y)$
is an SE-model of $\Pi$, since $Y \models \Pi$ implies that $Y \models \Pi^Y$.

Two logic programs (over the same atoms) are strongly equivalent iff they have
the same SE-models \cite{tocl:LifschitzPV01,tplp:Turner03}. Checking whether two logic programs
%have the same SE-models (and, therefore, if they are strongly equivalent) is
are strongly equivalent is
\co\NP-complete \cite{tplp:Turner03,tplp:PearceTW09}.

\paragraph{Epistemic Logic Programs.} An \emph{epistemic literal} is a formula
$\eneg \ell$, where $\ell$ is a literal and $\eneg$ is the epistemic negation
operator. A \emph{ground epistemic logic program (ELP)} is a triple $\Pi =
(\calA, \calE, \calR)$, where $\calA$ is a set of propositional atoms, $\calE$
is a set of epistemic literals over the atoms $\calA$, and $\calR$ is a set of
\emph{ELP rules}, which are
%rules 
%
%
\begin{equation*}
   a_1\vee \cdots \vee a_k \leftarrow \ell_1, \ldots, \ell_m, \xi_1, \ldots,
   \xi_j, \neg \xi_{j + 1}, \ldots, \neg \xi_{n},
\end{equation*}
where each $a_i\in\calA$ is an atom, each $\ell_i$ is a literal, and each $\xi_i
\in \calE$ is an epistemic literal. 
%Such rules are also called \emph{ELP rules}.
Note that usually $\calE$ is defined implicitly to be the set of all epistemic
literals appearing in the rules $\calR$; however, making the domain of epistemic
literals explicit will prove useful for our purposes.
%SW: not needed elsewhere:
%Similar to logic programs, let $\head{r} = \{ a_1, \ldots, a_k \}$, and let
%$\body{r} = \{ \ell_1, \ldots, \ell_m, \xi_1, \ldots, \xi_j, \neg \xi_{j+1},
%\ldots, \neg \xi_{n} \}$, that is, the set of elements appearing in the rule
%body. 
The \emph{union} %(or \emph{combination}) 
of two ELPs $\Pi_1 = (\calA_1,
\calE_1, \calR_1)$ and $\Pi_2 = (\calA_2, \calE_2, \calR_2)$ is the ELP $\Pi_1
\cup \Pi_2 = (\calA_1 \cup \calA_2, \calE_1 \cup \calE_2, \calR_1 \cup
\calR_2)$.

For a set $\calE$ of epistemic literals, a subset $\Phi \subseteq \calE$ of
epistemic literals is called an \emph{epistemic guess} (or, simply, a
\emph{guess}). The following definition provides a way to check whether a set of
interpretations is compatible with a guess~$\Phi$.

\begin{definition}\label{def:compatibility}
  Let $\calA$ be a set of atoms, $\calE$ be a set of epistemic literals over
  $\calA$, and $\Phi \subseteq \calE$ be an epistemic guess. A set $\calI$ of
  interpretations over $\calA$ is called \emph{$\Phi$-compatible w.r.t.\
  $\calE$}, iff 
  \begin{enumerate}
      \item\label{def:compatibility:1} $\calI \neq \emptyset$;
      \item\label{def:compatibility:2} for each epistemic literal $\eneg \ell
	\in \Phi$, there exists an interpretation $I \in \calI$ such that $I
	\not\models \ell$; and
      \item\label{def:compatibility:3} for each epistemic literal $\eneg \ell
	\in \calE \setminus \Phi$, for all interpretations $I \in \calI$ it
	holds that $I \models \ell$.
  \end{enumerate}
\end{definition}

For an ELP $\Pi = (\calA, \calE, \calR)$, the \emph{epistemic reduct}
\cite{ai:ShenE16} of the program $\Pi$ w.r.t.\ a guess $\Phi$, denoted $\Pi^\Phi$,
consists of the rules $\calR^\Phi=\{ r^\neg \mid r \in \calR \}$, where $r^\neg$ is defined
as the rule $r \in \calR$ where all occurrences of epistemic literals $\eneg
\ell \in \Phi$ are replaced by $\top$, and all remaining epistemic negation
symbols $\eneg$ are replaced by default negation $\neg$. Note that, after this
transformation, $\Pi^\Phi=(\calA,\calR^\Phi)$ is a logic program without epistemic
negation\footnote{In fact, $\Pi^\Phi$ may contain triple-negated atoms
$\neg\neg\neg a$. But, according to the definition of the GL-reduct in
\cite{amai:LifschitzTT99}, such formulas are equivalent to simple negated atoms
$\neg a$, and we treat them as such.}. This leads to the following, central
definition.

\begin{definition}\label{def:candidateworldview}
  Let $\Pi = (\calA, \calE, \calR)$ be an ELP. A set $\calM$ of interpretations
  over $\calA$ is a \emph{candidate world view (CWV)} of $\Pi$ if there is an
  epistemic guess $\Phi \subseteq \calE$ such that $\calM =
  \answersets{\Pi^\Phi}$ and $\calM$ is compatible with $\Phi$ w.r.t.\ $\calE$.
  The set of all CWVs of an ELP $\Pi$ is denoted by $\cwvs{\Pi}$. 
\end{definition}

Let us reconsider the CWA formulations as examples.

\begin{example}\label{ex:cwa-running1}
  Let $\calA_C=\{p,p'\}$, $\calE_C=\{\eneg p, \eneg \neg p\}$,
  $\Pi_G=(\calA_C,\calE_C,\calR_G)$ with $\calR_G$ containing
  only rule $(\ref{eq:gelfond-cwa})$, and
  $\Pi_S=(\calA_C,\calE_C,\calR_S)$ with $\calR_S$
  containing only rule $(\ref{eq:shen-eiter-cwa})$.

  We obtain $\cwvs{\Pi_G}=\{ \{ \{p'\} \} \}$ as guess $\Phi=\{ \eneg p \}$ is compatible with $\answersets{\Pi_G^{\Phi}=\{p' \gets \neg p\}} = \{\{p'\}\}$, while no other guesses are compatible with the answer sets of the respective epistemic reducts. We also obtain $\cwvs{\Pi_S}=\{ \{ \{p'\} \} \}$ as $\Phi$ is compatible with $\answersets{\Pi_S^{\Phi}=\{p' \gets \top \}} = \{\{p'\}\}$, while no other guesses are compatible with the answer sets of the respective epistemic reducts.
\end{example}

Following the principle of knowledge minimization, 
\citeauthor{ai:ShenE16} 
\shortcite{ai:ShenE16} 
define a
\emph{world view} as follows.
%as a CWV that has a maximal guess.

\begin{definition}\label{def:worldview}
  Let $\Pi = (\calA, \calE, \calR)$ be an ELP. %A CWV $\calC$ is called a
  $\calC\in\cwvs{\Pi}$ is called
  \emph{world view (WV)} of $\Pi$ if 
  its associated guess $\Phi$ is subset-maximal, i.e.\
  there is no $\calC'\in\cwvs{\Pi}$ with associated guess $\Phi'\supset \Phi$.
  %w.r.t.\ $\calE$.
  %
\end{definition}

Note that in Example~\ref{ex:cwa-running1} there is only one CWV per program; hence the
associated guesses are subset-maximal, and the sets of CWVs and WVs coincide.

One of the main reasoning tasks regarding ELPs is the \emph{world view existence
problem}, %(also called \emph{ELP consistency}), 
that is, given an ELP $\Pi$,
decide whether a WV (or, equivalently, CWV) exists. This problem is %known to be
\SIGMA{3}{P}-complete \cite{ai:ShenE16}.

%\section{Properties of Epistemic Logic Programs}\label{sec:properties}

\nop{
In this section, we will state several observations and properties about ELPs
that will help us to define and investigate the notion of strong equivalence for
ELPs. Several of these are not difficult to establish, but for the sake of
completeness we want to state them explicitly, as they will greatly simplify the
presentation of our main results in Section~\ref{sec:strongequivalence}.
}

We close this section with a statement that shows that extending $\calA$ or $\calE$ of an ELP does not change their CWVs (and hence also not their WVs).

\begin{theorem}\label{thm:atomelitdomains}
  Let $\Pi = (\calA, \calE, \calR)$ be an ELP and let $\Pi' = (\calA', \calE',
  \calR)$ be an ELP with the same set of rules, but with $\calA' \supset \calA$ and
  $\calE' \supset \calE$. Then, $\cwvs{\Pi} = \cwvs{\Pi'}$.
\end{theorem}

The proof can be found in Appendix~\ref{sec:thm:atomelitdomains:proof}.
%
%From the above, we can see that the set of CWVs of an ELP depends only on the
%actual rules and the atoms and epistemic literals occurring there, but not on
%the domain of atoms or epistemic literals. 
The above theorem implies that, given two
ELPs $\Pi_1 = (\calA_1, \calE_1, \calR_1)$ and $\Pi_2 = (\calA_2, \calE_2,
\calR_2)$, we can always assume that $\calA_1 = \calA_2$ and $\calE_1 = \calE_2$
(since the domains can be extended without changing the CWVs, as per the above
theorem). %which is very useful in the sequel.

%%Another interesting question is whether a program can always be modified in such
%%a way that all the atoms and epistemic literals from its respective domains
%%actually appear in the rules. We will investigate this question in the context
%%of tautological rules in Section~\ref{sec:casestudies}.

\section{Strong Equivalence for ELPs}\label{sec:strongequivalence}

In this section, we will investigate notions of equivalence of ELPs, in
particular, focusing on how to extend the concept of strong
equivalence~\cite{tocl:LifschitzPV01,tplp:Turner03} from logic programs to ELPs.
We will start by defining (ordinary) equivalence of two ELPs. %first, a
%straightforward extension of the same notion for logic programs. From the
%definition of the semantics of ELPs, we obtain two versions of equivalence,
%depending on whether we focus on WVs or CWVs.

\begin{definition}\label{def:equivalence}
  Two ELPs are \emph{WV-equivalent} (resp.\ \emph{CWV-equivalent}) iff their
  WVs (resp.\ CWVs) coincide.
\end{definition}

We observe that CWV-equivalence is the stronger notion, as it
immediately implies WV-equivalence. %Further, note that 
Moreover, for two ELPs to be
equivalent as defined above, not only must the set of guesses leading to WVs/CWVs
be the same, but also the answer sets in each of these WVs/CWVs. 
%w.r.t.\ these guesses.

%\begin{example}\label{ex:cwa-running2}
%  %
%  $\Pi_G$ and $\Pi_S$ of Example~\ref{ex:cwa-running1} are both CWV- and WV-equivalent.
%  %
%\end{example}

We now continue by defining strong equivalence for ELPs. One motivation for such
a kind of equivalence is modularization: we want to be able to replace a
sub-program by another one such that the semantics (i.e.\ WVs or CWVs) of the
overall program does not change. Based on the two equivalence notions defined
above and using ideas from work done in the area of logic programs \cite{amai:EiterFFW07},
we propose four notions of strong equivalence for ELPs.

\begin{definition}\label{def:strongequivalence}
  Let $\Pi_1$ and $\Pi_2$ be two ELPs. %For any third program $\Pi$,
  $\Pi_1$ and $\Pi_2$ are
  \begin{itemize}
    \item %$\Pi_1$ and $\Pi_2$ are 
	\emph{strongly ELP-WV-equivalent} iff, %$\Pi$ is an ELP, and 
      		for every ELP $\Pi$, 
	$\Pi_1 \cup \Pi$ and $\Pi_2 \cup \Pi$ are WV-equivalent;
    \item %$\Pi_1$ and $\Pi_2$ 
	\emph{strongly ASP-WV-equivalent} iff, %$\Pi$ is a (plain) logic program, and 
      for every (plain) logic program $\Pi$, 
	$\Pi_1 \cup \Pi$ and $\Pi_2 \cup \Pi$ are
      WV-equivalent;
    \item %$\Pi_1$ and $\Pi_2$ are 
	\emph{strongly ELP-CWV-equivalent} iff, %$\Pi$ is an ELP, and 
      		for every ELP $\Pi$, 
	  $\Pi_1 \cup \Pi$ and $\Pi_2 \cup \Pi$ are CWV-equivalent;
    \item %$\Pi_1$ and $\Pi_2$ are 
	\emph{strongly ASP-CWV-equivalent} iff, %$\Pi$
      for every (plain) logic program $\Pi$, 
	$\Pi_1 \cup \Pi$ and $\Pi_2 \cup \Pi$ are CWV-equivalent.
  \end{itemize}
\end{definition}

Having defined these equivalence notions for ELPs, the main aim of this section
is to characterize strong equivalence in a similar fashion as was done for logic
programs by \citeauthor{tplp:Turner03} \shortcite{tplp:Turner03}. Note that one
could be tempted to define strong equivalence for ELPs simply in terms of
\citeauthor{tplp:Turner03}'s SE-models of the epistemic reducts, for each
possible epistemic guess. However, this approach does not capture our
intentions, as the following example shows:

\begin{example}\label{ex:new}
  Take the two ELPs $\Pi_1 = (\calA, \calE, \calR_1)$ and $\Pi_2 = (\calA,
  \calE, \calR_2)$ with $\calR_1 = \{ p \gets \eneg p \}$, $\calR_2 = \{ p
  \gets \neg p \}$ %. Note that 
  and $\calE = \{ \eneg p \}$. Now, for the guess
  $\Phi_1 = \emptyset$, note that $\Pi_1^{\Phi_1} = \Pi_2^{\Phi_1}$ and hence,
  trivially, the SE-models are also the same. However, for the guess $\Phi_2 =
  \calE$, $\Pi_1^{\Phi_2}$ consists of the rule $p \gets \top$, while
  $\Pi_2^{\Phi_2}$ reduces to $p \gets \neg p$. It can be checked that the SE-models of these
  two epistemic reducts w.r.t.\ $\Phi_2$ are not the same, %since 
  hence these two
  programs are not strongly equivalent in the sense of \cite{tplp:Turner03}.
  However, it turns out that the guess $\Phi_2$ can never give rise to a CWV,
  since it requires that there is an answer set not containing $p$, but both 
  $\Pi_1^{\Phi_2}$ and $\Pi_2^{\Phi_2}$ require that $p$ is true in all answer
  sets of the CWV. Hence, it seems that the epistemic reducts for $\Phi_2$
  should not have any bearing on evaluating strong equivalence.
\end{example}

The example above implies that, when establishing strong equivalence for ELPs,
we should discard guesses that can never give rise to a CWV. We formalize this
as follows:

\begin{definition}\label{def:realizability}
  Let $\calI$ be a set of interpretations over a domain of atoms $\calA$, let
  $\calE$ be a set of epistemic literals over $\calA$, and $\Phi \subseteq
  \calE$ be a guess. Then, $\Phi$ is \emph{realizable} in $\calI$ iff there is a
  subset $\calI' \subseteq \calI$ such that $\calI'$ is compatible with $\Phi$ w.r.t.\ $\calE$.
\end{definition}

Given an ELP $\Pi = (\calA, \calE, \calR)$ and a guess $\Phi \subseteq \calE$,
we say that $\Phi$ is \emph{realizable in $\Pi$} iff $\Phi$ is realizable in the
set of models of $\Pi^\Phi$. We say that $\Phi$ is \emph{realizable in a set of
SE-models $\calS$} iff $\Phi$ is realizable in the set $\{ Y \mid (X, Y) \in
\calS \}$. %Based on the above definition, we 
We are now ready to define our central
construct, the SE-function $\calS\calE_\Pi$ of an ELP, which %, as we will see,
will %turn out to be a useful tool 
be the key concept to characterize strong equivalence for ELPs. Note that
realizability plays an important role in this.

\begin{definition}\label{def:sefunction}
  The \emph{SE-function} $\calS\calE_\Pi(\cdot)$ of an ELP $\Pi = (\calA, \calE, \calR)$ %is a function
  %$\calS\calE_\Pi$ that 
  maps guesses $\Phi \subseteq \calE$ for $\Pi$ to sets of SE-models as follows.
  %For guesses $\Phi \subseteq \calE$: 
$$\calS\calE_\Pi(\Phi) = \left\{
  \begin{array}{@{}lr@{}}
    \semods{\Pi^\Phi} & \text{if } \Phi \text{ realizable in } \Pi\\
    \emptyset & \text{otherwise.}
  \end{array}
  \right.$$
\end{definition}

\begin{example}\label{ex:cwa-running3}
  Recall that programs $\Pi_G$ and $\Pi_S$ of Example~\ref{ex:cwa-running1} are both CWV- and WV-equivalent.
  %Reconsidering $\Pi_G$ and $\Pi_S$ from Example~\ref{ex:cwa-running2}, 
  However, we find that
  their respective SE-functions differ: %For example, it can be checked that
  $\calS\calE_{\Pi_G}(\{ \eneg p \})$ contains the tuple $(\emptyset,\{
  p \})$, but $\calS\calE_{\Pi_S}(\{ \eneg p \})$
  does not; yet $\{ \eneg p \}$ is realizable in both $\Pi_G$ and $\Pi_S$.
\end{example}

Note that if $\Phi$ is realizable in $\Pi$, then $\{ Y \mid (X, Y) \in
\calS\calE_\Pi(\Phi) \} = \mods{\Pi^\Phi}$. Before proceeding to our main
results, we first state some observations that can be made about the SE-function
of an ELP.

\begin{lemma}\label{lem:worldviewfromsefunc}
  Let $\Pi = (\calA, \calE, \calR)$ be an ELP with $\calS\calE_\Pi$ its
  SE-function.  Further, let $\calM$ be a set of interpretations, and $\Phi
  \subseteq \calE$ be a guess. Then, $\calM$ is a CWV of $\Pi$ w.r.t.\ $\Phi$
  iff $\calM = \{ Y \mid (Y, Y) \in \calS\calE_\Pi(\Phi), \neg\exists X \subset
  Y \ (X, Y) \in \calS\calE_\Pi(\Phi) \}$ and $\calM$ is compatible with $\Phi$.
\end{lemma}

\begin{proof}
  The left-to-right direction can be shown as follows. By definition of the
  SE-function it holds that for each CWV $\calM$ w.r.t.\ a guess $\Phi$ for
  $\Pi$, $\calS\calE_\Pi(\Phi)$ contains the set $\{ (Y, Y) \mid Y \in \calM
  \}$, and, since $\calM$ contains only answer sets of the epistemic reduct
  $\Pi^\Phi$, there cannot be a pair $(X, Y)$ in $\calS\calE_\Pi(\Phi)$ with $X
  \subset Y$. %as, by definition of SE-models, this would contradict that $Y$ is
  %an answer set of $\Pi^\Phi$. 
  There also cannot be some other pair $(Y', Y')
  \in \calS\calE_\Pi(\Phi)$ for which no pair $(X', Y')$ with $X' \subset Y'$ exists, since
  then this would mean that $Y'$ is an answer set of $\Pi^\Phi$, and therefore
  must be in $\calM$. Finally, since by assumption $\calM$ is a CWV, it is
  compatible with $\Phi$ by definition.

  For the right-to-left direction, note that %, by definition of the SE-function,
  %the set 
  %$\calM = \{ Y \mid (Y, Y) \in \calS\calE_\Pi(\Phi), \neg\exists X
  %\subset Y \ (X, Y) \in \calS\calE_\Pi(\Phi) \}$ 
  $\calM$ clearly contains all those
  sets of atoms $Y$ %such that $Y$ is a 
  that are models of %the reduct 
  $\Pi^\Phi$, %and there
  %is no subset of $X$ that is a 
  such that there is no subset of $Y$ model of the GL-reduct $[\Pi^\Phi]^Y$. Hence,
  $\calM$ contains precisely the answer sets of $\Pi^\Phi$. Since, by
  assumption, $\calM$ is compatible with %guess 
  $\Phi$, this immediately implies
  that $\calM$ is a CWV of $\Pi$ as per Definition~\ref{def:candidateworldview}.
\end{proof}

The next statement is a direct consequence of the previous lemma, since, as we
have seen, the SE-function of an ELP defines its CWVs.

\begin{lemma}\label{lem:sefuncequivalence}
  Programs with the same SE-function are CWV-equivalent (and hence
  WV-equivalent).
\end{lemma}

We are now ready to state the main result of this section, namely that the
SE-function precisely characterizes strong equivalence for ELPs.

\begin{theorem}\label{thm:strongequivalence}
  Let $\Pi_1$ and $\Pi_2$ be two ELPs. The following statements are equivalent:
  \begin{enumerate}
    \item\label{thm:strongequivalence:1} $\Pi_1$ and $\Pi_2$ are
      ELP-CWV-equivalent;
    \item\label{thm:strongequivalence:2} $\Pi_1$ and $\Pi_2$ are
      ASP-CWV-equivalent;
    \item\label{thm:strongequivalence:3} $\Pi_1$ and $\Pi_2$ are
      ELP-WV-equivalent;
    \item\label{thm:strongequivalence:4} $\Pi_1$ and $\Pi_2$ are
      ASP-WV-equivalent; and
    \item\label{thm:strongequivalence:5} $\calS\calE_{\Pi_1} = \calS\calE_{\Pi_2}$.
  \end{enumerate}
\end{theorem}

%Before giving the proof, let us introduce the shorthand $\square S$ to denote
%the conjunction $\bigwedge_{a \in S} \square a$, for some set $S$ of atoms and a
%unary operand $\square$, which shall be the identity function when not present.
%We will use this notation to turn sets of atoms into (parts of) rule bodies.

\begin{proof}
  \textbf{$(\ref{thm:strongequivalence:1}) \Rightarrow
  (\ref{thm:strongequivalence:2}) \Rightarrow (\ref{thm:strongequivalence:4}),
  (\ref{thm:strongequivalence:1}) \Rightarrow (\ref{thm:strongequivalence:3})
  \Rightarrow (\ref{thm:strongequivalence:4})$.} These follow directly from
  Definition~\ref{def:strongequivalence} and from the fact that every WV is a
  CWV and every ASP program is an ELP.

  \smallskip\noindent\textbf{$(\ref{thm:strongequivalence:5}) \Rightarrow
  (\ref{thm:strongequivalence:1})$.} 
  %Given Lemmas~\ref{lem:worldviewfromsefunc}
  %and~\ref{lem:sefuncequivalence}, it is not difficult to show this direction.
  Assume that statement $(5)$ holds. %that is, that the two ELPs $\Pi_1$ and
  %$\Pi_2$ have the same SE-function. 
  We need to show that for any third program
  $\Pi$ it holds that $\Pi_1 \cup \Pi$ and $\Pi_2 \cup \Pi$ are equivalent. To
  this end, pick any guess $\Phi$. %It is easy to verify that 
  Then,
  $$
  \calS\calE_{\Pi_1 \cup \Pi}(\Phi) =$$ $$\left\{\begin{array}{@{}lr@{}}
	\calS\calE_{\Pi_1}(\Phi) \cap \calS\calE_\Pi(\Phi) \qquad & \text{if }
	\Phi \text{ is realizable in }\\
	& \calS\calE_{\Pi_1}(\Phi) \cap \calS\calE_\Pi(\Phi)\\[2ex]
	\emptyset & \text{otherwise}
  \end{array}\right\}$$ $$= \calS\calE_{\Pi_2 \cup \Pi}(\Phi),$$ since
  $\calS\calE_{\Pi_1}(\Phi) = \calS\calE_{\Pi_2}(\Phi)$ by assumption, and thus
  $\calS\calE_{\Pi_1}(\Phi) \cap \calS\calE_\Pi(\Phi)  =
  \calS\calE_{\Pi_2}(\Phi) \cap \calS\calE_\Pi(\Phi)$.
  Lemma~\ref{lem:sefuncequivalence} then proves that $\Pi_1 \cup \Pi$ and $\Pi_2
  \cup \Pi$ are equivalent.

  \smallskip\noindent\textbf{$(\ref{thm:strongequivalence:4}) \Rightarrow
  (\ref{thm:strongequivalence:5})$.} We will prove the contrapositive. Let
  $\Pi_1 = (\calA_1, \calE_1, \calR_1)$ and $\Pi_2 = (\calA_2, \calE_2,
  \calR_2)$ be two ELPs and assume, w.l.o.g., that $\calA_1 = \calA_2 = \calA$
  and $\calE_1 = \calE_2 = \calE$ (cf.\ Theorem~\ref{thm:atomelitdomains}).
  Further, let $\Phi \subseteq \calE$ be a guess such that $\calY_1 =
  \calS\calE_{\Pi_1}(\Phi) \neq \calS\calE_{\Pi_2}(\Phi) = \calY_2$. Finally,
  assume that there is a pair $(X, Y)$ in $\calY_1$ but not in $\calY_2$ (again,
  w.l.o.g., by symmetry). We need to show that there exists a logic program
  $\Pi$ (i.e.\ without epistemic literals) such that the WVs of the ELP $\Pi_1
  \cup \Pi$ differ from those of $\Pi_2 \cup \Pi$. %(i.e.\ such that the two
  %combined programs are not WV-equivalent). Further, we 
  We only need to consider
  the case where $\Pi_1$ and $\Pi_2$ are WV-equivalent, since the claim is
  trivially true otherwise. %simply let $\Pi$ be the empty program.

  By Definition~\ref{def:sefunction}, with $\calY_1$ non-empty by assumption,
  there is a subset $\calC \subseteq \{ Y \mid (X, Y) \in \calY_1 \}$ compatible
  with $\Phi$, since $\Phi$ is realizable in $\Pi$. Let $\calC = \{ Y_1,
  \ldots, Y_m \}$ and let $\{ Y_{m+1}, \ldots, Y_n \} = 2^\calA \setminus
  \calC$. %that is, the set of all other possible interpretations. Note that $n = 2^{|\calA|}$.

  The idea is to construct $\Pi$ in such a way that the potential WV represented
  by $\calC$ is actually realized in $\Pi_1 \cup \Pi$. To construct $\Pi$, let
  $y_1, \ldots, y_n$ be fresh atoms not occurring in $\calA$. Let $\Pi$ contain
  the rule $y_1 \vee \ldots \vee y_n \gets \top,$ and, furthermore, for all $1
  \leq i \leq n$ and $a \in Y_i$, the rule $\bot \gets y_i, \neg a$, and for all
  $a \in \calA \setminus Y_i$, the rule $\bot \gets y_i, a$. 
  %Finally, for each
  %such $i$, add the rule $\bot \gets \neg y_i, Y_i, \neg (\calA \setminus Y_i)$.
  This makes sure that for every model $Y_i$ of $\Pi_1^\Phi$, the corresponding
  model $Y_i'$ of $(\Pi_1 \cup \Pi)^\Phi$ contains the atom $y_i$ (i.e.\ $Y_i' =
  Y_i \cup \{ y_i \}$).

  Now, take the pair $(X, Y)$ that is, by assumption, contained in $\calY_1$
  (but not in $\calY_2$). Clearly, there is some integer $k$, $1 \leq k \leq n$,
  such that $Y = Y_k$. Now, for each model $Y_i$, $1 \leq i \leq m$, $i \neq k$,
  and each atom $a \in Y_i$, add the rule $a \gets y_i$ to $\Pi$. For each model
  $Y_i$, $m < i \leq n$, $i \neq k$, add the rule $\bot \gets y_i$ to $\Pi$.
  This makes sure that, in $(\Pi_1 \cup \Pi)^\Phi$, exactly the models from
  $\calC$ (except $Y_k$, if $k \leq m$) become answer sets, and all other models
  %from $\calY_1$ and $\calY_2$ 
  are destroyed, except for the model $Y = Y_k$, if
  $k > m$.
  
  At this point, if we have that $\{ Y_1, \ldots, Y_m \} \not\subseteq
  \mods{\Pi_2^\Phi}$, we simply do the same as above also for the model $Y =
  Y_k$ (i.e.\ realize $Y_k$ as an answer set if $k \leq m$, or destroy $Y_k$ in
  case $k > m$). Then, clearly, $(\Pi_1 \cup \Pi)^\Phi$ will have the answer
  sets $\{ Y_1 \cup \{ y_1 \}, \ldots, Y_m \cup \{ y_m \} \}$ which form a CWV
  of $\Pi_1 \cup \Pi$, but not of $\Pi_2 \cup \Pi$. Since all other models are
  destroyed, independent of the guess $\Phi$, this CWV is actually the only CWV
  of $\Pi_1 \cup \Pi$, and hence, it is a WV, proving our claim that $\Pi_1 \cup
  \Pi$ and $\Pi_2 \cup \Pi$ are not WV-equivalent \textbf{($\star$)}. It therefore remains to show
  the claim for the case where the set $\{ Y_1, \ldots, Y_m \} \subseteq
  \mods{\Pi_2^\Phi}$. To this end, we need to distinguish the following two
  cases:

  \smallskip \noindent \textit{Case (1).} Assume that $Y \not\models
  \Pi_2^\Phi$. In this case, for each atom $a \in Y$, add the rule $a \gets y_k$
  to $\Pi$. %In this case, it is not difficult to verify that 
  Now, $Y \cup \{ y_k \}$
  is an answer set of $(\Pi_1 \cup \Pi)^\Phi$, but not of $(\Pi_2 \cup
  \Pi)^\Phi = \Pi_2^\Phi \cup \Pi$. 

  \smallskip \noindent \textit{Case (2).} Assume that $Y \models \Pi_2^\Phi$. In
  this case, for each atom $a \in X$, add the rule $a \gets y_k$ to $\Pi$, and,
  in addition, for all atoms $a, b \in Y \setminus X$, add the rule $a \gets b,
  y_k$ to $\Pi$. We will show that, in this case, $Y$ is an answer set of
  $(\Pi_2 \cup \Pi)^\Phi$, but not of $(\Pi_1 \cup \Pi)^\Phi$. Since $Y \models
  \Pi_2^\Phi$ and every model of a program is also a model of its GL-reduct, by
  definition of SE-models we know that $Y \neq X$, since, by assumption, $(X, Y)
  \in \calY_1$ but $(X, Y) \not\in \calY_2$. Since $(X, Y) \in \calY_1$, we have
  that $X \models [\Pi_1^\Phi]^Y$. But then, by construction of $\Pi$ it holds
  that $X \cup \{ y_k \} \models [(\Pi_1 \cup \Pi)^\Phi]^Y = [\Pi_1^\Phi]^Y \cup
  [\Pi^\Phi]^Y$ and therefore $Y \cup \{ y_k \}$ is not an answer set of $(\Pi_1
  \cup \Pi)^\Phi$. On the other hand, for $(\Pi_2 \cup \Pi)^\Phi$, assume that
  there is some $X' \subset Y$ such that $X' \cup \{ y_k \} \models [(\Pi_2 \cup
  \Pi)^\Phi]^Y$. Clearly, by construction of $\Pi$, $X \subset X'$. Thus, there
  is some atom $a$ in $X' \subseteq Y$ but not in $X$. But, by construction of
  $\Pi$, we then have that $X' = Y$. Hence, $Y \cup \{ y_k \}$ is an answer set
  of $(\Pi_2 \cup \Pi)^\Phi$, as desired.

  \smallskip The above shows that, for both cases (1) and (2), the set $\calC' =
  \{ Y_1 \cup \{ y_1 \}, \ldots, Y_m \cup \{ y_m \} \}$ is the set of answer
  sets of one of the two programs $(\Pi_1 \cup \Pi)^\Phi$ and $(\Pi_2 \cup
  \Pi)^\Phi$, and therefore, by assumption, a CWV of that program.  But, by
  construction of $\Pi$, $\calC'$ cannot be a CWV for the other program (because
  $Y \cup \{ y_k \}$ is an answer set of one of the two programs, but not both,
  and hence distinguishes the CWVs)\footnote{Note that, in particular, in case
  (1), $\calC'$ is a CWV of $\Pi_1 \cup \Pi$ but not of $\Pi_2 \cup \Pi$, and in
  case (2), the same holds if $k \leq m$, and the reverse holds if $k > m$.}. It
  remains to show that $\calC'$ is not only a CWV but also a WV of exactly one
  of the two programs, which can be done via the argument for
  \textbf{($\star$)}.  This concludes the proof.
\end{proof}

The above theorem states that the SE-function precisely characterizes all the
notions of strong equivalence for ELPs and that these notions are all
equivalent. We therefore will, from now on, jointly refer to these four
equivalent notions as \emph{strong equivalence}. The
next statement follows immediately.

\begin{corollary}\label{corr:strongequivalence}
  ELPs are strongly equivalent iff they have the same SE-function.
\end{corollary}

\begin{example}\label{ex:cwa-running4}
  Continuing from Example~\ref{ex:cwa-running3}, we observe that the difference in
  the SE-function can be made manifest by combining the two programs with the
  program $\Pi = (\calA_C, \calE_C, \calR)$, where $\calR = \{ p \gets \neg p'\}$. We observe that $\cwvs{\Pi_G \cup \Pi} = \{\{\{p\}\}\}$ (due to guess $\{ \eneg p \}$) and $\cwvs{\Pi_S \cup \Pi} = \{\{\{p'\}\}\}$ (due to guess $\{ \eneg \neg p \}$).
\end{example}

We close this section by %comparing our notion of strong equivalence for
%ELPs with the original definition of strong equivalence for plain
%logic programs. 
%From the above lemma and
%Theorem~\ref{thm:strongequivalence}, %we obtain the following
%corollary, which states
%if follows 
showing that our definition of strong equivalence for
ELPs generalizes the established notion of strong equivalence for
logic programs.

\begin{corollary}\label{corr:segeneralization}
  Let $\Pi_1 = (\calA, \calR_1)$ and $\Pi_2 = (\calA, \calR_2)$ be two logic
  programs, and let $\Pi_1' = (\calA, \emptyset, \calR_1)$ and $\Pi'_2 = (\calA,
  \emptyset, \calR_2)$ be two ELPs with the same sets of rules as $\Pi_1$ and
  $\Pi_2$, respectively. %and hence with $\calE = \emptyset$. 
   Then, $\Pi_1$ and
  $\Pi_2$ are strongly equivalent iff $\Pi_1'$ and $\Pi_2'$ are.
\end{corollary}
\begin{proof}(Sketch)
For $\Pi'_1$ and $\Pi'_2$
there is only one possible guess, %$\Phi
%  \subseteq \calE$, 
namely, $\Phi = \emptyset$. 
Moreover, in this setting it holds that 
 $\calS\calE_{\Pi}(\emptyset) =
  \semods{\Pi}$ for any ELP $\Pi$: 
  the crucical observation is that $\semods{\Pi}$
  is compatible with %the empty guess 
  $\Phi$ w.r.t.\ 
  the empty domain of epistemic literals, exactly
  if $\semods{\Pi}\neq\emptyset$.
\end{proof}

\section{Complexity of ELP Strong Equivalence}\label{sec:complexity}

%Since we have now seen that two programs are strongly equivalent iff they have
%the same SE-function, the question of the computational complexity of deciding
%strong equivalence arises.
We now make use of our characterization to settle the computational complexity
of deciding strong equivalence between ELPs.
%We start with the following lemma, which 
The following lemma shows that to check the realizability
of a guess $\Phi$ for a given ELP $\Pi$ it suffices to consider a
polynomially-sized subset of the models of $\Pi$.

\begin{lemma}\label{lem:polynomialwitness}
  Let $\Pi = (\calA, \calE, \calR)$ be an ELP, $\Phi \subseteq \calE$ a guess
  for $\Pi$, and $\calY = \calS\calE_\Pi(\Phi)$ be non-empty. Then, there is a
  set 
  $\calC \subseteq \{ Y \mid (X, Y) \in \calY \}$
  %$\calY' \subseteq \calY$ 
  of polynomial size in $\Pi$ %, such that the set
  %$\{ Y \mid (X, Y) \in \calY' \}$ 
  that is compatible with $\Phi$ w.r.t. $\calE$.
\end{lemma}

\begin{proof}
  According to Definition~\ref{def:sefunction}, since $\Phi$ is realizable in
  $\Pi$, there is a subset $\calY' \subseteq \calY$, such that the set $\calC'
  = \{Y \mid (X, Y) \in \calY' \}$ is compatible with $\Phi$.  %But, according
  By
  Definition~\ref{def:compatibility}, for each %epistemic literal 
  $\eneg \ell
  \in \calE$, when $\Phi$ contains $\eneg \ell$, there is a $Y_{\eneg \ell}
  \in\calC'$ such that $\ell$ is false in $Y_{\eneg \ell}$. Take the subset
  $\calC \subseteq \calC'$ containing $Y_{\eneg \ell}$ for each
  %epistemic literal 
  $\eneg \ell \in \calE$, or, if $\Phi = \emptyset$, let
  $\calC$ be any singleton subset of $\calC'$. Note that $\calC$ is of
  polynomial size in $\Pi$. Clearly, $\calC$ is also compatible with $\Phi$.
  %Let $\calY'$ be the subset of $\calY''$ containing, for each $Y \in \calC'$,
  %precisely one pair from $\calY''$ having $Y$ as its second component.
  %
\end{proof}

With this crucial observation in place, we are now ready to state the complexity
of checking strong equivalence for ELPs,  %Interestingly, the complexity 
which remains
in \co\NP\ as for plain logic programs. %even though the complexity for checking strong equivalence of logic
%programs has the same complexity. 
This is surprising, since %the classical
reasoning tasks for ELPs generally are one level higher on the polynomial
hierarchy than the corresponding %reasoning 
tasks for logic programs (cf.\
Section~\ref{sec:preliminaries}).

\begin{theorem}\label{thm:secomplexity}
  Checking whether two ELPs are strongly equivalent is \co\NP-complete.
\end{theorem}

\begin{proof}
  %
  %We will show the theorem by
  %exhibiting 
  We give a non-deterministic polynomial time procedure for checking that two
  ELPs are \emph{not} strongly equivalent, that is, that there is a difference
  in their respective SE-functions. W.l.o.g. let $\Pi_1 = (\calA, \calE, \calR_1)$
  and $\Pi_2 = (\calA, \calE, \calR_2)$. %be such two ELPs. We assume,
  %w.l.o.g., that $\calA_1 = \calA_2 = \calA$ and $\calE_1 = \calE_2 = \calE$.
  The procedure works as follows:
  \begin{enumerate}
    \item Guess an epistemic guess $\Phi \subseteq \calE$.
    \item\label{proof:secomplexity:step2} Guess a polynomially-sized, non-empty
      set of interpretations $\calC$ over $\calA$, compatible with $\Phi$.
    \item Verify in polynomial time that each $Y \in \calC$ is a model of both
      $\Pi_1^\Phi$ and $\Pi_2^\Phi$. If not, REJECT.
    \item Guess a %candidate SE-model 
          pair $(X, Y)$ with $X\subseteq Y\subseteq \calA$.
    \item Verify in polynomial time that $(X, Y) \in \calS\calE_{\Pi_1}(\Phi)$
      but $(X, Y) \not\in \calS\calE_{\Pi_2}(\Phi)$ or vice versa. If so,
      ACCEPT. If not, REJECT.
  \end{enumerate}

  Clearly, the above procedure terminates in polynomial time, since it is known
  that model checking for SE-models can be done in polynomial time \cite{amai:EiterFFW07}.
  To obtain correctness, it is not difficult to verify that the above procedure
  is sound and complete given the following two observations. Firstly, note that
  Lemma~\ref{lem:polynomialwitness} implies that we only need to focus on
  polynomially sized subsets when evaluating the realizability of a guess as
  stated in Definition~\ref{def:sefunction}; hence guessing a polynomially sized
  set of interpretations is enough in step~\ref{proof:secomplexity:step2}.
  Secondly, assume that two programs $\Pi_1$ and $\Pi_2$ have differing
  SE-functions. Then this means that there must be some guess $\Phi$, such that
  both sets $\calS\calE_{\Pi_1}(\Phi)$ and $\calS\calE_{\Pi_2}(\Phi)$ are
  non-empty (clearly, as otherwise $\calS\calE_{\Pi_1}(\Phi) =
  \calS\calE_{\Pi_2}(\Phi) = \emptyset$ for all guesses $\Phi$). But, by
  Definition~\ref{def:sefunction}, this means that both sets contain a potential
  CWV. Now there are two cases: either they do not share any potential CWVs, or
  if they do, then there is at least one SE-model $(X, Y)$ contained in one but
  not both of the two sets.

  Corollary~\ref{corr:segeneralization} allows us to inherit the \co\NP\ lower
  bound from the case of logic programs \cite{tplp:PearceTW09}, which completes the proof.
\end{proof}

\section{Case Studies}\label{sec:casestudies}

In this section, we apply our characterisation 
to investigate basic principles for simplifying ELPs.

%Having completed our investigation of the properties and complexity of strong
%equivalence for ELPs in general, in this section we will now examine several
%cases where these notions are useful, specifically in terms of ELP
%simplification. We will start by looking at the concept of tautological (also known as 
%redundant) rules.

\subsection{Tautological Rules}\label{sec:tautologicalrules}

Tautological rules are rules that can simply be removed from any program 
without affecting its outcome.
%where
%they occur without changing the CWVs of the program, as can be seen from the
%following definition.

%%%%%%%%%%%%%%%%%%%%%%%%%%%%%%%%%%%%%%%%%%%%%%%%%%%%%%%%%%%%%%%%%%%%%%%%
%% TODO (journal?): show that we generalize tautological rules for LPs %
%%%%%%%%%%%%%%%%%%%%%%%%%%%%%%%%%%%%%%%%%%%%%%%%%%%%%%%%%%%%%%%%%%%%%%%%

\begin{definition}\label{def:tautologicalrule}
  An ELP rule $r$ is \emph{tautological} iff the single-rule ELP $\Pi_r =
  (\calA, \calE, \{ r \})$ is strongly equivalent to the empty program
  $\Pi_\emptyset = (\calA, \calE, \emptyset)$.
\end{definition}

%We will see in the following that we can use our notion of the SE-function to
%characterize tautological rules. To this end, let $\calA$ be a set of atoms and
Let $\calE$ be a set of epistemic literals over $\calA$. 
We say that an
epistemic guess $\Phi \subseteq \calE$ is \emph{consistent} iff, whenever
$\calE$ contains both $\eneg a$ and $\eneg \neg a$ for some atom $a \in \calA$,
it holds that $\Phi$ contains at least one of $\{ \eneg a, \eneg \neg a \}$.\footnote{%
Clearly, if $\Phi$ contains neither of the two epistemic literals, then $\Phi$
can never be realizable in any ELP, since there is no (non-empty) set of models
where $a$ is both always true and always false.} 
%Now, towards our characterization, 
Moreover, let $\calS_\calA$ denote the set of all pairs $(X, Y)$ such
that $X\subseteq Y\subseteq \calA$.

\begin{lemma}\label{lem:tautologicalrulesefunc}
  An ELP rule $r$ is tautological iff for the single-rule ELP $\Pi_r = (\calA,
  \calE, \{ r \})$ it holds that $\calS\calE_{\Pi_r}(\Phi) = \calS_\calA$ for
  each consistent guess $\Phi \subseteq \calE$.
\end{lemma}

\begin{proof}
  This follows from the fact that for the empty program $\Pi_\emptyset =
  (\calA,\calE,\emptyset)$, $\calS\calE_{\Pi_\emptyset}(\Phi) = \calS_\calA$ for
  each consistent guess $\Phi \subseteq \calE$, and the observation that if
  $\semods{\Pi^\Phi}=\calS_\calA$, then $\Phi$ is realizable in $\Pi$.
\end{proof}

Before syntactically characterizing tautological ELP rules, 
we recall a corresponding result for standard ASP rules from the literature.
%let us do this for
%standard rules of form (\ref{eq:rule}) first\footnote{
%  %
%  Such an analysis has been done in \cite{jair:LinC07} for example, but without
%  double negation in the rules.
%  %
%}. 
For convenience, we shall denote ASP rules of the form (\ref{eq:rule})
as
\begin{equation}\label{eq:rulewithsets}
  A \leftarrow B, \neg C, \neg\neg D,
\end{equation}
using sets of atoms $A$, $B$, $C$, and $D$ 
as is common practice.
In what follows, we denote single-rule logic programs consisting of 
a rule $r$ by
$\Pi_r = (\calA,\calR)$ with $\calR=\{r\}$
and  call $r$ tautological if $\Pi_r$ is strongly equivalent to 
the empty program $\Pi=(\calA,\emptyset)$, i.e.\
$\semods{\Pi_r}=S_{\calA}$.
%SW: might be confusing for $A$  ...
%are sets of atoms, and $\square S$ denotes the
%conjunction $\bigwedge_{a \in S} \square a$, for some set $S$ of atoms and a
%unary operand $\square$.  
%
%Note that for such a rule $r$ and an interpretation
%$I$, the GL-reduct of the single-rule program $\Pi_r = (\calA, \calR)$, with
%$\calR = \{ r \}$, is the program $\Pi_r^I = (\calA, \calR^I)$, with $\calR^I =
%\{ A \gets B \mid C \cap Y \neq \emptyset, D \subseteq Y \}$.
%
The 
following result %carries over from
is due to \cite[Lemma~2]{iclp:CabalarPV07}. 
%using Corollary~\ref{corr:segeneralization}.

\begin{lemma}\label{lem:lpruletautcharacterization}
  A rule $r$ of the form (\ref{eq:rulewithsets}) %, $\Pi_r$ 
  is tautological iff ($\alpha$) $A
  \cap B \neq \emptyset$, ($\beta$) $B \cap C \neq \emptyset$, or ($\gamma$) $C \cap D \neq
  \emptyset$.
\end{lemma}

%% \begin{proof}
  %% %
  %% Let $\Pi_r = (\calA, \calR)$ be the single-rule program with $\calR = \{ r
  %% \}$.
%% 
  %% \smallskip\noindent\textbf{($\Rightarrow$)} Let $r$ be a rule such that (a),
  %% (b), or (c) holds. It is easy to check that for each $Y \subseteq \calA$ we
  %% have that $Y \models \Pi_r$. Now, let $X \subseteq Y$. In case (a), clearly,
  %% $X \models \Pi_r^Y$; for (b) observe that if $\calR^Y \neq \emptyset$, then $X
  %% \not \models B$ since some $a \in B \cap C$ is not in $Y$, and hence $X
  %% \models \Pi_r^Y$; and for (c) note that $\calR^Y = \emptyset$, and thus $X
  %% \models \Pi_r^Y$ for each $Y \subseteq \calA$, since either $C \cap Y \neq
  %% \emptyset$ or $D \not\subseteq Y$. It thus follows that any pair $(X,Y)$,
  %% where $X \subseteq Y \subseteq \calA$, is an SE-model of $\Pi_r$.
  %% 
  %% \smallskip\noindent\textbf{($\Leftarrow$)} Let $r$ be any rule such that none of
  %% (a)--(c) hold, and consider the pair $(X,Y)$ with $X=B$ and $Y=B\cup D$. We
  %% will show that $(X,Y)$ is not an SE-model of $\Pi_r$.  If $Y \not\models r$ we
  %% are done, so suppose $Y \models r$. It follows that $(A \cup C) \cap (B \cup
  %% D) \neq \emptyset$. Since we cannot have (a)--(c), it must be the case $A \cap
  %% D \neq \emptyset$.  By construction and since $C \cap (B \cup D) = \emptyset$,
  %% we have $\calR^Y = \{ A \gets B \}$.  Now, since $A \cap B = \emptyset$, we
  %% get that $X \not\models \Pi_r^Y$.
  %% %
%% \end{proof}

We are now ready to characterize tautological ELP rules. For the sake of
presentation let us write them as 
\begin{equation}\label{eq:elprulewithsets}
  A \leftarrow B, \neg C, \eneg D, \eneg\neg E, \neg\eneg F, \neg\eneg\neg G,
\end{equation}
where, again, each capital letter represents a set of atoms, analogously to
ASP rules of the form (\ref{eq:rulewithsets}).
Note that $D$ here plays a different role than in 
(\ref{eq:rulewithsets}); this is due
to the fact that in ELP rules we have no explicit double negation.

Let us also consider the notion of epistemic reduct w.r.t.\ the above notation.
Analogously to the %logic program 
ASP case, let $\Pi_r = (\calA, \calE, \calR)$ be
the single-rule ELP with $\calR = \{ r \}$.  Now, for %some guess 
$\Phi \subseteq
\calE$, we have $\Pi_r^\Phi = (\calA, \calE, \calR^\Phi)$ with $\calR^\Phi =
\emptyset$ if $\eneg f \in \Phi$ for some $f \in F$ or if $\eneg \neg g \in
\Phi$ for $g \in G$; otherwise, 
$$
\calR^\Phi  = \{ r^\Phi \} = \{ A \gets B,\neg C, \neg D^\Phi,
\neg\neg E^\Phi, \neg\neg F, \neg G \} 
$$ 
with $D^\Phi=\{ d\in D, \eneg d \notin
\Phi\}$, $E^\Phi=\{ e\in E, \eneg \neg e \notin \Phi\}$. We are now ready to give
our full syntactical characterization of tautological ELP rules. 
It shows that deciding whether a rule is tautological amounts to a
simple syntactic check and can be done individually, for each rule.
%(proof in the supplementary material).

\begin{theorem}\label{thm:elpruletautcharacterization}
  %
  %For an ELP rule $r$ of the form (\ref{eq:elprulewithsets}), $\Pi_r$
  An ELP rule $r$ of form (\ref{eq:elprulewithsets})
is tautological iff (a) $A \cap B \neq \emptyset$, (b) $B \cap
  (C \cup G) \neq \emptyset$, (c) $C \cap F \neq \emptyset$, (d) $D \cap F \neq
  \emptyset$, (e) $E \cap G \neq \emptyset$, or (f) $F\cap G\neq\emptyset$.
\end{theorem}

\begin{proof}
  Let $\Pi_r = (\calA, \calE, \calR)$ be the single-rule ELP with $\calR = \{ r
  \}$, and, for an epistemic guess $\Phi$, let $r^\Phi$ denote the unqiue rule in $\calR^\Phi$ in
  case $\calR^\Phi\neq\emptyset$.

  \smallskip\noindent\textbf{($\Rightarrow$)} 
  Let $r$ be an ELP rule satisfying at least one of the
  conditions (a)--(f), and let $\Phi \subseteq \calE$ be a consistent guess. %We have to
  By Lemma~\ref{lem:tautologicalrulesefunc}, we have to
  show that either $\calR^\Phi = \emptyset$ or that $r^\Phi$
  fulfills the conditions of 
  %where $r^\neg$ is tautological in the sense of
  Lemma~\ref{lem:lpruletautcharacterization}. 
  This can be easily verified for
  conditions (a)--(c). For (d), note that if $\eneg f\in\Phi$ for some $f\in F$
  we get $\calR^\Phi = \emptyset$; otherwise for all $f \in F$ it holds that
  $\eneg f \notin \Phi$, and thus that $D \cap F \neq \emptyset$, and it follows
  that $f \in D^\Phi$ for some $f \in F$. Hence we have $\neg f$ and $\neg\neg f$ in
  $r^\Phi$, which is tautological as per
  Lemma~\ref{lem:lpruletautcharacterization}, condition ($\gamma$).   
  The argument is similar for (e),
  %Similarly for (e), note that if
  %$\eneg \neg g \in \Phi$ for some $g \in G$, we get $\calR^\Phi = \emptyset$;
  %otherwise, for all $g \in G$, $\eneg \neg g \notin \Phi$, and thus that $E
  %\cap G \neq \emptyset$, and it follows that $g \in E^\Phi$ for some $g \in G$.
  %Again we have $\neg g$ and $\neg\neg g$ in $r^\Phi$, fulfilling condition
  %($\gamma$) of Lemma~\ref{lem:lpruletautcharacterization}.
  Finally, for (f) note that since $\Phi$ is consistent, it must contain 
  one of $\eneg a$ and $\eneg \neg a$ for each
  $a\in F\cap G$. It follows that $\calR^\Phi=\emptyset$.

  \smallskip\noindent\textbf{($\Leftarrow$)} Let $r$ be an ELP rule such that none
  of (a)--(f) hold.  Let $\Phi = \calE \setminus
  (\eneg F \cup \eneg \neg G)$. $\Phi$ is consistent since (f) does not hold. Moreover,  
  $\eneg D \cup \eneg \neg E \subseteq \Phi$
  (since neither (d) nor (e) holds). 
  Suppose, $\Phi$ is realizable in
  $\Pi_r$. %we are done, since 
  %$\Phi$ is realizable. We will 
 % Otherwise, we will show that $\calS\calE_{\Pi_r}(\Phi)
 % \neq \calS_\calA$.  
%
By construction of $\Phi$, 
$\calR^\Phi = %\{ r^\neg \} = 
  \{ A \gets B, \neg C, \neg\neg F, \neg G \}$; by
  Lemma~\ref{lem:lpruletautcharacterization} and the fact that ($\alpha$)--($\gamma$) of that
  lemma do not hold for $r^\Phi$, we have that %the program $\Pi_{r^\Phi}$ containing 
  $r^\Phi$ is not tautological, i.e.\ $\semods{\Pi_{r^\Phi}}\neq S_{\calA}$, and thus
 % Otherwise, we will show that 
	$\calS\calE_{\Pi_r}(\Phi) \neq \calS_\calA$.  
Otherwise, 
  we get $\calS\calE_{\Pi_r}(\Phi) = \emptyset$. %contradicting Lemma~\ref{lem:tautologicalrulesefunc}. Thus,
Hence, for both cases, %we get that
  $r$ is not tautological by Lemma~\ref{lem:tautologicalrulesefunc}.
\end{proof}

%The above theorem completes our characterization of tautological rules in the
%context of ELPs. It says that checking that a rule is tautological amounts to a
%simple syntactic check and can be done individually, for each rule. 
%In fact,
%given a rule $r$, this check can clearly be done in \LOGSPACE\ in the size of
%$r$, and can further be shown to be in \AC{0}.

\subsection{Rule Subsumption}

Rule subsumption identifies when a rule $s$ can safely be removed
from a program $\Pi$,
given another rule $r$ is in $\Pi$.

\begin{definition}
  An ELP rule $s$ is \emph{subsumed} by an ELP rule $r$ iff %the single-rule ELP
  %$\Pi_r = (\calA,\calE,\{r\})$ 
  $\Pi_r$ is strongly equivalent to %the combined program
  $\Pi_r \cup \Pi_s$. %where $\Pi_s$ is the single-rule ELP given by
  %$(\calA,\calE,\{r,s\})$.
  %
\end{definition}

The next result follows from
 the definition of the SE-function for a union of two ELPs (cf.\ proof of
 Theorem~\ref{thm:strongequivalence}).

\begin{lemma}\label{lemma:sub}
  An ELP rule $s$ is subsumed by an ELP rule $r$ iff, 
  for %the single-rule 
  ELPs
  $\Pi_r = (\calA, \calE, \{ r \})$ and $\Pi_s = (\calA, \calE, \{ s \})$, it
  holds that 
  $\calS\calE_{\Pi_r}(\Phi) \subseteq \calS\calE_{\Pi_s}(\Phi)$, for
  all guesses $\Phi \subseteq \calE$.
\end{lemma}

%\begin{proof}
%  %
%  We need to show that $\Pi_r$ is strongly equivalent to $\Pi_r \cup \Pi_s$ iff
%  for all guesses $\Phi \subseteq \calE$, $\calS\calE_{\Pi_r}(\Phi)
%  \subseteq \calS\calE_{\Pi_s}(\Phi)$. However, this %, however, %can be easily seen 
%  follows from
%  the definition of the SE-function for a union of two ELPs (cf.\ proof of
%  Theorem~\ref{thm:strongequivalence}).
%
  %such that 
  %$\Phi$ is realizible in $R^\Phi$,
  %$\semods{R^\Phi}\subseteq\semods{S^\Phi}$.
  %{\bf TBF.}
  %Only-if: Let $\Phi\subseteq \calE$. 
  %We have $\calS\calE_\Pi(R)=\calS\calE_\Pi(S)$. If $\Phi$ is realizable in $R^\phi$, 
  %$\semods{R^\Phi}=\semods{(R\cup S)^\Phi}$ and thus $\semods{R^\Phi}\subseteq\semods{S^\Phi}$.
  %
%\end{proof}

Clearly, a tautological rule $s$ is subsumed by any other rule, hence 
in what follows we focus on subsumption of non-tautological rules only. 
Again, we %start our analysis with 
exploit known results from ASP.
With some abuse of terminology, we say that an ASP rule $r$
subsumes another such rule $s$
iff $\Pi_r=(\calA,\{r\})$ is strongly equivalent to $\Pi_{r,s}=(\calA,\{r,s\})$,
i.e.\ iff
$\semods{\Pi_r}\subseteq\semods{\Pi_s}$.
%
%The following result is equivalent to 
We first adapt a result from
\cite[Theorem~7]{iclp:CabalarPV07}
to our notation.

\begin{lemma}\label{lemma:subrule}
An ASP rule 
$r=A\leftarrow B,\neg C,\neg\neg D$ 
subsumes a non-tautological ASP rule 
$s=A'\leftarrow B',\neg C',\neg\neg D'$  
iff the following conditions jointly hold:
($\alpha$) $A\subseteq A'\cup C'$,
($\beta$) $B\subseteq B'\cup D'$,
($\beta$') if $A\cap (A'\setminus C')\neq\emptyset$ then $B\subseteq B'$,
($\gamma$) $C\subseteq C'$, and
($\delta$) $D\subseteq B'\cup D'$.
\end{lemma}

We can now give a syntactic criterion for
ELP rule subsumption, which turns out to be somewhat involved, 
but still feasible to check.
It requires
two technical notions that link a rule $r$ to a rule $s$ whenever
$r$ has sufficiently many elements that are not ``absorbed'' by default-negated epistemic
literals in $s$.

%\begin{definition}
%For the two ELP rules
%$r=A \leftarrow B, \neg C, \eneg D, \eneg\neg E, \neg\eneg F, \neg\eneg\neg G$
%and 
%$s=A' \leftarrow B', \neg C', \eneg D', \eneg\neg E', \neg\eneg F', \neg\eneg\neg G'$
%we write 
%$r\rhd s$ 
%if $|(A\cup C\cup D)\setminus G'|>1$
%or
%$(B\cup E)\setminus F'\neq \emptyset$.
%and 
%$r\RHD s$ 
%if $(A\cup C\cup D)\setminus G'\neq\emptyset$
%or 
%$|(B\cup E)\setminus F'|>1$.
%
%is called proper
%if
%$|A\cup C\cup D|>1$ or
%$B\cup E\neq \emptyset$.
%$|A\cup C\cup D|>1$ or
%$B\cup E\neq \emptyset$.
%$|A\cup C\cup D\cup G|>1$ or
%$B\cup E\cup F\neq \emptyset$.
%\end{definition}

We are now ready to give our syntactic characterization.

\begin{theorem}\label{thm:sub}
Let 
$r$ be an ELP-rule of form (\ref{eq:elprulewithsets})
%= A \leftarrow B, \neg C, \eneg D, \eneg\neg E, \neg\eneg F, \neg\eneg\neg G$
and 
$s =A' \leftarrow B', \neg C', \eneg D', \eneg\neg E', \neg\eneg F',
 \neg\eneg\neg G'$ be non-tautological.
%
%\begin{eqnarray*}
%r&\!\!=\!\!& A \leftarrow B, \neg C, \eneg D, \eneg\neg E, \neg\eneg F, \neg\eneg\neg
%  G \\
%  s &\!\!=\!\!&A' \leftarrow B', \neg C', \eneg D', \eneg\neg E', \neg\eneg F',
%  \neg\eneg\neg G'
%\end{eqnarray*}
%be two ELP rules, $s$ being non-tautological,  
Furthermore, let $r\rhd s$ 
denote that $|(A\cup C\cup D)\setminus G'|>1$
or
$(B\cup E)\setminus F'\neq \emptyset$,
and 
$r\RHD s$  denote that
$(A\cup C\cup D)\setminus G'\neq\emptyset$
or 
$|(B\cup E)\setminus F'|>1$.
%\smallskip

\noindent
Then, 
$r$ subsumes $s$
iff 
the following conditions jointly hold:
\begin{enumerate}
%\item[(a)] $A\subseteq A'\cup C'\cup D'\cup G'$,
\item[(a)] $A\subseteq A'\cup C'\cup D'\cup G'$;
\item[(a$^*$)] if $r\rhd s$ then $A\subseteq A'\cup C'\cup G'$;
%\item[(b)] $B\subseteq B'\cup E'\cup F'$,
\item[(b)] $B\subseteq B'\cup E'\cup F'$;
\item[(b$^*$)] if $r\RHD s$ then $B\subseteq B'\cup F'$;
\item[(b')] if $A \cap (A'\setminus (C' \cup D' \cup G'))\neq\emptyset$ then $B\subseteq B'$;
\item[(b$^*$')] if $r\rhd s$ and $A\cap (A'\setminus (C' \cup G'))\neq\emptyset$ then $B\subseteq B'$;
%\item[(b')] if $A\cap (A'\setminus (C'\cup D' \cup G'))\neq\emptyset$ then $B\subseteq B'$,
%\item[(b$^*$')] if $A\cap (A'\setminus (C'\cup G'))\neq\emptyset$ then $B\subseteq B'$,
%\item[(c)] $C\cup D \subseteq C'\cup D'\cup G'$, 
\item[(c)] 
%$C\subseteq C'\cup G'$ and
$C\cup D\subseteq C'\cup D'\cup G'$;
\item[(c$^*$)] 
 if %$r$ is proper 
 $r\rhd s$
 then $C\subseteq C'\cup G'$;
\item[(d)] $E\subseteq B'\cup E'\cup F'$;
\item[(e)] $F\subseteq F'$ and  $G\subseteq G'$.
\end{enumerate}
\end{theorem}

Note that, in the above, items (a)--(d) generalize the same items in
Lemma~\ref{lemma:subrule}, and the proof for those parts is a generalization of
the proof of that lemma. The full proof can be found in
Appendix~\ref{sec:thm:sub:proof}. 

 %Take the two ELPs $\Pi_1 = (\calA, \calE, \calR_1)$ and $\Pi_2 = (\calA,
 % \calE, \calR_2)$ with 
Interestingly, applying the above theorem shows, for example, that  
$r= p \gets \eneg p$ subsumes $s= p \gets \neg p $ and vice versa 
(in particular, since neither $r\rhd s$ nor $s\rhd r$), showing that 
the two programs in Example~\ref{ex:new} are strongly equivalent.

\section{Conclusions}\label{sec:conclusions}

In this paper, a simple characterization of strong equivalence for epistemic logic programs was proposed, which also demonstrates that various notions of strong equivalence coincide. The characterization generalizes strong equivalence for plain logic programs, and, somewhat unexpectedly, shows that checking strong equivalence for ELPs is %as difficult as 
not harder than for ASP in terms of computational complexity. The results also give rise to a syntactic characterization of tautological ELP rules and ELP rule subsumption. 

As another byproduct, we studied the relationship between two formalizations of
CWA, Gelfond-CWA and Shen-Eiter-CWA, as our running example. Indeed, while they
are (ordinarily) equivalent, they are not strongly equivalent, as shown in
Examples~\ref{ex:cwa-running3} and \ref{ex:cwa-running4}. In particular,
Example~\ref{ex:cwa-running4} shows that, combining $\Pi$ in that example with the
Shen-Eiter-CWA rule, yields the world view $\{ \{ p' \} \}$, which does not seem
intuitive in this setting. However, in \cite{ai:ShenE16}, it seems that the CWA
rule is proposed for Reiter's CWA \cite{adbt:Reiter77}, and thus was not intended to work with rules
containing default negation.

For future work, 
%we want to apply the notions of subsumption and tautological rules to 
we want to apply our findings to
obtain a normal form for ELPs, 
as was done for ASP in \cite{iclp:CabalarPV07}.
Furthermore,
we plan to study weaker forms of equivalence \cite{tocl:EiterFW07} for ELPs. In particular, 
it will be interesting to see whether our notion of SE-functions can be similarly re-used for 
characterizing uniform equivalence like SE-models did serve as a basis for UE-models.

\section*{Acknowledgements} Michael Morak and Stefan Woltran were supported by
the Austrian Science Fund (FWF) under grant number Y698.

\appendix

\section{Proof of Theorem~\ref{thm:atomelitdomains}}
\label{sec:thm:atomelitdomains:proof}

We
start by re-stating a folklore result from the world of ASP, namely that the
universe of atoms of a logic program can be extended without changing its answer
sets.

\begin{proposition}\label{prop:atomdomainlp}
  Let $\Pi = (\calA, \calR)$ be a logic program and let $\Pi' = (\calA', \calR)$
  be a logic program with the same set of rules, but with $\calA' \supset
  \calA$.  Then, $\answersets{\Pi} = \answersets{\Pi'}$.
\end{proposition}

It is easy to see that the above proposition holds by noting that any atom $a
\in \calA' \setminus \calA$ can clearly not appear in the rules $\calR$, and
hence any such $a$ must be false in any answer set of $\Pi'$ (i.e.\ for all $M
\in \answersets{\Pi'}$ it holds that $a \not\in M$). From this result, it is
easy to obtain a similar result for ELPs.

\begin{proposition}\label{prop:atomdomainelp}
  Let $\Pi = (\calA, \calE, \calR)$ be an ELP and let $\Pi' = (\calA', \calE,
  \calR)$ be an ELP with the same set of rules and domain of epistemic literals,
  but with $\calA' \supset \calA$. Then, $\cwvs{\Pi} = \cwvs{\Pi'}$.
\end{proposition}

\begin{proof}
  Note that, clearly, any atom $a \in \calA' \setminus \calA$ cannot appear
  anywhere in $\calR$ or $\calE$. Hence, for every guess $\Phi \subseteq \calE$,
  $a$ also does not appear in the rules of the epistemic reduct $\Pi'^\Phi$ and
  thus also not in any answer set of $\Pi'^\Phi$. But then,
  $\answersets{\Pi'^\Phi} = \answersets{\Pi^\Phi}$.
\end{proof}

From the above, we can see that the atom domain of an ELP can be arbitrarily
extended without changing the CWVs of the ELP. We will now show that the same is
the case for the domain of epistemic literals.

\begin{proposition}\label{prop:elitdomain}
  Let $\Pi = (\calA, \calE, \calR)$ be an ELP and let $\Pi' = (\calA, \calE',
  \calR)$ be an ELP with the same set of rules and atom domain, but with $\calE'
  \supset \calE$ and $\calE' \setminus \calE = \{ \eneg \ell \}$. Then,
  $\cwvs{\Pi} = \cwvs{\Pi'}$.
\end{proposition}

\begin{proof}
  Note that the epistemic literal $\eneg \ell$ cannot appear anywhere
  in $\calR$, but the atom in $\ell$ can be in $\calA$ and appear in $\calR$. Now consider any guess $\Phi \subseteq \calE$, which we will relate to $\Phi' = \Phi \cup \{\eneg \ell\}$. The rules of the
  epistemic reducts $\Pi^\Phi$, $\Pi'^\Phi$, and $\Pi'^{\Phi'}$ coincide, thus $\answersets{\Pi^\Phi} =
  \answersets{\Pi'^{\Phi'}} = \calM$. 
  First, consider the case that $\calM$ is not compatible with $\Phi$, it is not hard to see that then $\calM$ is not compatible with $\Phi'$ either.
  % I. $\calM=\emptyset$ is trivial.
  % II. There is $\eneg \ell' \in \Phi$ such that $\forall M \in \calM: M \models \ell'$; clearly $\eneg \ell' \in \Phi'$ as well, so $\calM$ is not compatible with $\Phi'$.
  % III. There is $\eneg \ell' \in \calE \setminus \Phi$ such that $\exists M \in \calM: M \not\models \ell'$; also here the same $\eneg \ell'$ is in $\calE' \setminus \Phi'$ as well, so $\calM$ is not compatible with $\Phi'$.
  Now assume that $\calM$ is compatible with $\Phi$, and thus $\calM \in \cwvs{\Pi}$. We distinguish two cases: (1) $\forall M \in \calM: M \models \ell$, in this case $\Phi$ is compatible with $\calM$ (with respect to $\calE'$) and therefore $\calM \in \cwvs{\Pi'}$; (2) $\exists M \in \calM: M \not\models \ell$, in this case $\Phi'$ is compatible with $\calM$ and thus $\calM \in \cwvs{\Pi'}$.
\end{proof}

From Propositions~\ref{prop:atomdomainelp} and~\ref{prop:elitdomain} in that
order, and their respective proofs, Theorem~\ref{thm:atomelitdomains} follows.

\section{Proof of Theorem~\ref{thm:sub}}
\label{sec:thm:sub:proof}

  Let $\Pi_r = (\calA, \calE, \calR_r)$ and $\Pi_s = (\calA, \calE, \calR_s)$ be
  the two single rule ELPs with $\calR_r = \{ r \}$ and $\calR_s = \{ s \}$, and
  let $r^\Phi$ and $s^\Phi$ be as in the previous section (cf.\
  the paragraph before Theorem~\ref{thm:elpruletautcharacterization}).

\smallskip
\noindent
\textbf{($\Leftarrow$)} 
Let $r,s$ be such that (a)--(e) hold. By Lemma~\ref{lemma:sub} we need to 
show 
that $\calS\calE_{\Pi_r}(\Phi) \subseteq \calS\calE_{\Pi_s}(\Phi)$, for
  all %consistent 
guesses $\Phi \subseteq \calE$.
If $\Phi$ is not consistent $\calS\calE_{\Pi_r}(\Phi) \subseteq \calS\calE_{\Pi_s}(\Phi)$ holds by trivial means.
Hence, let $\Phi \subseteq \calE$ be consistent. 
If $\Pi_s^\Phi=\emptyset$, we get $\calS\calE_{\Pi_s}(\Phi)=\calS_\calA$ 
and are done. 
So suppose $\Pi_s^\Phi\neq\emptyset$. 
By (e) and definition of the epistemic reduct, it follows that 
$\Pi_r^\Phi\neq\emptyset$. 
We next show that 
either $\Phi$ is not realizable in $\Pi_r$ 
(hence, 
$\calS\calE_{\Pi_r}(\Phi)=\emptyset$ and
$\calS\calE_{\Pi_r}(\Phi) \subseteq \calS\calE_{\Pi_s}(\Phi)$ obviously holds)
or 
$r^\Phi$ and $s^\Phi$
%$\Pi_r^\Phi$ and $\Pi_s^\Phi$ 
apply to conditions 
($\alpha$)--($\delta$) of Lemma~\ref{lemma:subrule}. 
In the latter case, 
due to this lemma then 
$\semods{\Pi_r^\Phi}\subseteq \semods{\Pi_s^\Phi}$,
and it remains to show that  $\calS\calE_{\Pi_r}(\Phi) \subseteq \calS\calE_{\Pi_s}(\Phi)$.
%We can now apply Lemma~\ref{lemma:subrule} to
%$\Pi_r^\Phi$ and $\Pi_s^\Phi$ and derive that
%$\semods{\Pi_r^\Phi}\subseteq \semods{\Pi_s^\Phi}$. It
%remains to show that  $\calS\calE_{\Pi_r}(\Phi) \subseteq \calS\calE_{\Pi_s}(\Phi)$. 
Towards a contradiction suppose this inclusion does not hold. In the light
of $\semods{\Pi_r^\Phi}\subseteq \semods{\Pi_s^\Phi}$, it must be the case
that 
$\calS\calE_{\Pi_r}(\Phi)=\semods{\Pi_r^\Phi}$ and 
$\calS\calE_{\Pi_s}(\Phi)=\emptyset$, i.e.\ $\Phi$ is not realizable in 
$\Pi_s^\Phi$. By definition this means that there is no subset of $\mods{\Pi_s^\Phi}$ that is compatible with $\Phi$. 
However, since $\semods{\Pi_r^\Phi}\subseteq\semods{\Pi_s^\Phi}$, 
it also follows that $\mods{\Pi_r^\Phi}\subseteq\mods{\Pi_s^\Phi}$.
Hence, $\Phi$ cannot be realizable in $\Pi_r^\Phi$ either and thus
$\calS\calE_{\Pi_r}(\Phi)=\emptyset$; a contradiction.

We continue to show that either $\Phi$ is not realizable in $\Pi_r$,
or 
%$\Pi_r^\Phi$ and $\Pi_s^\Phi$ 
$r^\Phi$ and $s^\Phi$ 
fulfill the conditions
($\alpha$)--($\delta$) of Lemma~\ref{lemma:subrule}.
Recall that we assume that $r^\Phi$ and $s^\Phi$ exist, i.e.\ the epistemic reducts are not empty.
Hence, 
$\Phi\subseteq \calE \setminus (\eneg F' \cup \eneg \neg G')$.
In fact, the epistemic reducts are of the form
$r^\Phi=
  A \leftarrow B, \neg C, \neg D^\Phi, \neg\neg E^\Phi, \neg\neg F, \neg G$
and 
$s^\Phi=
  A' \leftarrow B', \neg C', \neg D'^\Phi, \neg\neg E'^\Phi, \neg\neg F', \neg G'$, 
  where for $\Delta \in \{ D, D' \}$ we let
$\Delta^\Phi=\{d\in \Delta\mid \eneg d\notin \Phi\}$, and for
$\Gamma\in\{E,E'\}$ we let
$\Gamma^\Phi=\{e\in \Gamma\mid \eneg \neg e\notin \Phi\}$.
%Observe that 
%for $d\in D^\Phi\cap D'^\Phi$, $d\in D^\Phi$ implies $d\in D'^\Phi$; and
%analogously for $E^\Phi$ and $E'^\Phi$.

We first show that  condition ($\delta$) of Lemma~\ref{lemma:subrule} holds 
for $r^\Phi$ and $s^\Phi$.
Observe that for each $e\in E^\Phi$ with $e\in E'$, also $e\in E'^\Phi$. 
Hence from (d) we get
$E^\Phi\subseteq B'\cup E'^\Phi\cup F'$
and using $F\subseteq F'$ from (e)
we observe
$E^\Phi\cup F\subseteq B'\cup E'^\Phi\cup F'$ as desired.

We continue with condition ($\gamma$) of Lemma~\ref{lemma:subrule}.  Suppose it 
does not hold for $r^\Phi$ and $s^\Phi$.
We have to show that $\calS\calE_{\Pi_r}(\Phi)=\emptyset$.
Since the condition does not hold for $r^\Phi$ and $s^\Phi$, we have
$C\cup D^\Phi\cup G\not\subseteq C'\cup D'^\Phi\cup G'$.
Since $G\subseteq G'$ by (e),  
$C\cup D^\Phi\subseteq C'\cup D'^\Phi\cup G'$. Note that 
$D^\Phi\subseteq C'\cup D'^\Phi\cup G'$  follows from 
(c) and the fact that for each $d\in D^\Phi$ with $d\in D'$, also $d\in D'^\Phi$. 
Thus $C\not \subseteq C'\cup D'^\Phi\cup G'$. 
Let $c\in C\setminus (C'\cup D'^\Phi\cup G')$.
Since we assume by (c) $C \subseteq C'\cup D'\cup G'$, 
we have two observations:
(i) $\eneg c\in \Phi$ (otherwise $c\in D'^\Phi$);
(ii) due to condition (c$^*$), we need that
$r\not\rhd s$, since otherwise $C\subseteq C'\cup G'$ would hold.
From $r\not\rhd s$, it follows that there is neither an element in $(A\cup C\cup D) \setminus G'$ that is different from $c$ 
nor is there an element in $B\cup E$ different from $F'$.
We show that no subset of models of $r^\Phi$
is $\Phi$-compatible. 
First, observe that we would need for each such model that
all $f\in F'$ are set ot true and all $g\in G'$  are set to false.
Since $\eneg c \in \Phi$, we need among those models one that sets $c$ to false.
However, such a model does not exist:
the only ``positive'' atoms in $r^\Phi$ besides $c$ occur also in $G'$ and are already set to false;
all ``negative'' atoms in $r^\Phi$ are also in $F'$ and thus are set to true.
Hence,  $\calS\calE_{\Pi_r}(\Phi)=\emptyset$.

Next, we treat condition ($\beta$) of Lemma~\ref{lemma:subrule}.  Suppose it 
does not hold for $r^\Phi$ and $s^\Phi$.
Again, we need to show $\calS\calE_{\Pi_r}(\Phi)=\emptyset$.
Since the condition does not hold for $r^\Phi$ and $s^\Phi$, we have
$B \not\subseteq B'\cup E'^\Phi\cup F'$.
Let $b\in B\setminus (B'\cup E'^\Phi\cup F')$.
Since we assume by (b) $B \subseteq B'\cup E'\cup F'$, 
we have two observations:
(i) $\eneg \neg b\in \Phi$ (otherwise $b\in E'^\Phi$);
(ii) due to condition (b$^*$), we need 
$r\not\RHD s$ since otherwise $B\subseteq B'\cup F'$ would hold.
From $r\not\RHD s$, it follows that there is neither an element in $(B\cup E) \setminus F'$ that is different from $b$ 
nor is there an element in $A\cup C\cup D$ different from $G'$.
We show that no subset of models of $r^\Phi$
is $\Phi$-compatible. First, observe that we would need for each such model that
all $f\in F'$ are set ot true and all $g\in G'$  are set to false.
Since $\eneg \neg b \in \Phi$, we need among those models one that sets $b$ to true.
However, such a model does not exist:
all ``positive'' atoms in $r^\Phi$ occur also in $G'$ and are already set to false;
the only ``negative'' atoms in $r^\Phi$ different to $b$ are also in $F'$ and thus are set to true.
Hence,  $\calS\calE_{\Pi_r}(\Phi)=\emptyset$.

For Condition ($\beta$') of Lemma~\ref{lemma:subrule}, 
once more suppose
that is does not hold for $r^\Phi$ and $s^\Phi$. We have 
$X=A\cap (A'\setminus (C'\cup D'^\Phi\cup G'))\neq \emptyset$ and 
$B\not \subseteq B'$. 
%Let $a\in X$. If $a\in D'$, we know that $a\notin D'\Phi$, hence in that case $\not a \in \Phi$.
It follows that 
$A\cap (A'\setminus (C'\cup G'))\neq \emptyset$, 
and, since (b$^*$') holds, $r\not\rhd s$.
On the other hand, since (b') holds, 
we know that for each $x\in X$, $x\in D'$ and, hence $\eneg x \in \Phi$ (since $X$ does not contain
an element from $D'^\Phi$).
Fix some $a\in X$.
%(i) $\eneg \neg b\in \Phi$ (otherwise $b\in E'^\Phi$);
From $r\not\rhd s$ it follows that there is neither an element in $(A\cup C\cup D) \setminus G'$ that is different from $a$ 
nor is there an element in $B\cup E$ different from $F'$.
We show that no subset of models of $r^\Phi$ 
is $\Phi$-compatible.
First, observe that we would need for each such model that
all $f\in F'$ are set ot true and all $g\in G'$  are set to false.
Since $\eneg a \in \Phi$, we need among those models one that sets $a$ to false.
However, such a model does not exist:
the only ``positive'' atoms in $r^\Phi$ besides $a$ occur also in $G'$ and are already set to false;
all ``negative'' atoms in $r^\Phi$ are also in $F'$ and thus are set to true.
Hence,  $\calS\calE_{\Pi_r}(\Phi)=\emptyset$.

It remains to show that condition ($\alpha$) of Lemma~\ref{lemma:subrule} applies
to $r^\Phi$ and $s^\Phi$. The argument is similar to the one for condition ($\gamma$).
So, suppose the condition does not hold for $r^\Phi$ and $s^\Phi$,
%We have to show that $\calS\calE_{\Pi_r}(\Phi)=\emptyset$.
%Since the condition does not hold for $r^\Phi$ and $s^\Phi$, 
i.e.\ we have
$A\not\subseteq A'\cup C'\cup D'^\Phi\cup G'$.
Let $a\in A\setminus (A\cup C'\cup D'^\Phi\cup G')$.
Due to assumption (a),  
(i) $\eneg a\in \Phi$ (otherwise $a\in D'^\Phi$);
(ii) $r\not\rhd s$,
since otherwise $A\subseteq A'\cup C'\cup G'$ via  assumption (a$^*$).
From $r\not\rhd s$, it follows that there is neither an element in $(A\cup C\cup D) \setminus G'$ that is different from $a$ 
nor is there an element in $B\cup E$ different from $F'$.
Showing that no subset of models of $r^\Phi$ 
is $\Phi$-compatible
follows essentially the same
arguments as used for condition (c) above.
%is $\Phi$-compatible. First, observe that we would need for each such model that
%all $f\in F'$ are set ot true and all $g\in G'$  are set to false.
%Since $\eneg a \in \Phi$, we need among those models one that sets $a$ to false.
%However, such a model does not exist:
%the only ``positive'' atoms in $r^\Phi$ besides $a$ occur also in $G'$ and are already set to false;
%all ``negative'' atoms in $r^\Phi$ are also in $F'$ and thus are set to true.
Hence, we end up with  $\calS\calE_{\Pi_r}(\Phi)=\emptyset$ as desired. This concludes 
the proof of the only-if direction.

\medskip
\noindent\textbf{($\Rightarrow$)} 
Suppose one of the conditions (a)--(e) is violated.
We show that there exists a guess $\Phi$, such that 
$\calS\calE_{\Pi_r}(\Phi) \not\subseteq \calS\calE_{\Pi_s}(\Phi)$.

Let us first consider 
(e) is violated and let
$\Phi=\calE\setminus (\eneg F' \cup \eneg \neg G')$.  
We get $r^\Phi=\emptyset$ and thus $\calS\calE_{\Pi_r}(\Phi)=\calS_\calA$, while
  $s^\Phi=A' \leftarrow B', \neg C', %\neg D'^*, \neg\neg E'^*, 
\neg\neg F', \neg G'$ (note that $D'$ and $E'$ dissappear since we assume
$s$ to be non-tautological).
%with 
%$D'^*\subseteq  D'$
%and
%$E'^*\subseteq  E'$. 
It can be checked that, since $s$ is non-tautological, 
so is $s^\Phi$, and thus $\calS\calE_{\Pi_s}(\Phi)\subset \calS_\calA$. 
For the remaining cases, let us assume (e) holds. 
Note that then 
$r^\Phi$ is non-empty as well. In fact, it is of the form
$r^\Phi=
  A \leftarrow B, \neg C, \neg D^\Phi, \neg\neg E^\Phi, \neg\neg F, \neg G$
with 
$D^\Phi=\{d\in D\mid \eneg d\notin \Phi\}$ 
and 
$E^\Phi=\{e\in E\mid \eneg \neg e\notin \Phi\}$.

First, suppose (d) is violated.
We consider $\Psi=\Phi\setminus \{\eneg \neg e\}$ with 
$e\in E\setminus (B'\cup E'\cup F')$.  
Since $\Psi\subseteq \Phi$, $r^\Psi$ and $s^\Psi$ are non-empty.
In fact, $s^\Psi=s^\Phi$ since $e\notin E'$.
Moreover in $r^\Psi$ we have $e\in E^\Psi$. 
We observe 
that condition ($\delta$) of Lemma~\ref{lemma:subrule} is violated,
and, from the same lemma,
we get that $\semods{\Pi^\Psi_r}\not\subseteq \semods{\Pi_s^\Psi}$. 
It remains to show that $\mods{\Pi^\Psi_r}$ realizes $\Psi$.
Note that any interpretation that sets $e$ to false is a model 
of $r^\Psi$.  
We thus take $\calI$  as the set of interpretations 
where $e$ is set to false,
all $f\in F'$ are set to true, and 
all $g\in G'$ are set to false.
Indeed, $\calI\subseteq \mods{r^\Psi}$ since 
$\calI$ is a set of well-formed interpretations 
(recall that  $e \notin F'$, and $F'\cap G'=\emptyset$ since we assumed
$s$ to be non-tautological).
We show that the three conditions 
for $\calI$ being $\Psi$-compatible hold.
(i) $\calI\neq\emptyset$ by definition;
(ii) for any $\eneg \ell \in \Psi$, there exists 
an $I\in\calI$, such that $I\not\models \ell$:
by definition $\ell\notin \{\neg e\} \cup F' \cup \neg G'$.
For $\ell = e$, $I\not\models e$ hold for all $I\in\calI$.
Similary this is true for $\ell\in \neg F'$ and for all $\ell\in G'$.
%Clearly, $\ell\neq \neg e$. 
For all other $\ell$ such that $\eneg \ell \in \Psi$, 
there is one $I\in\calI$ such that $I\not\models\ell$
by construction of $\calI$;
(iii) we show that for any $\eneg \ell \in \calE\setminus \Psi$, 
and for each each $I\in \calI$, 
$I\models \ell$. Indeed, by construction of $\Psi$, 
$\ell$ is either $\neg e$, $f\in F'$ or $\neg g\in \neg G'$;
by definition,
for
each $I\in\calI$ it holds that $I \models \ell$.

Suppose (c) is violated.
We consider $\Psi=\Phi\setminus \{\eneg c\}$ with 
$c\in (C\cup D)\setminus (C'\cup D'\cup G')$.  
Since $\Psi\subseteq \Phi$, $r^\Psi$ and $s^\Psi$ are non-empty.
In fact, $s^\Psi=s^\Phi$ since $c\notin D'$.
Moreover in $r^\Psi$ we have $c\in D^\Psi$, for the case $c\notin C$. 
Thus, $c\in C\cup D^\Psi$.
We observe 
that condition ($\gamma$) of Lemma~\ref{lemma:subrule} is thus violated,
and, by the same lemma,
we get that $\semods{\Pi_r^\Psi}\not\subseteq \semods{\Pi_s^\Psi}$. 
It remains to show that $\Psi$ is realizable in $\mods{\Pi_r^\Psi}$.
Indeed, any interpretation that sets $c$ to true is a model 
of $r^\Psi$.  
We take $\calI$  as the set of interpretations 
where
$c$ is set to true,
all $f\in F'$ are set to true, and 
all $g\in G'$ are set to false. 
As before, one can show that $\calI$ is a set of 
well-formed interpretations (in paricular, since $c\notin G'$) 
and
$\Psi$-compatible.
%Indeed, $\calI\subseteq \mods{r^\Psi}$.
%We show that the three conditions 
%for $\calI$ being $\Psi$-compatible hold.
%(i) $\calI\neq\emptyset$ is clear;
%(ii) for any $\eneg \ell \in \Psi$, there exists 
%an $I\in\calI$, such that $I\not\models \ell$:
%by definition $\ell\notin \{c\} \cup F' \cup \neg G'$.
%For $\ell = \neg c$, $I\not\models \neg c$ hold for all $I\in\calI$.
%Similary this is true for $\ell\in \neg F'$ and for all $\ell\in G'$.
%For all other $\ell$ such that $\eneg \ell \in \Psi$, 
%there is one $I\in\calI$ such that $I\not\models\ell$
%by construction of $\calI$;
%(iii) we show that for any $\eneg \ell \in \calE\setminus \Psi$, 
%and for each each $I\in \calI$, 
%$I\models \ell$. Indeed, by construction of $\Psi$, 
%$\ell$ is either $c$, $f\in F'$ or $\neg g\in \neg G'$. 
%However, each $I\in\cal I$ satisfies those $\ell$.

Suppose (c$^*$) is violated, i.e.\
we have 
$|(A\cup C\cup D)\setminus G'|>1$ or
$(B\cup E)\setminus F'\neq \emptyset$,
but
$C\not\subseteq C'\cup G'$.
Let $c\in C\setminus(C'\cup G')$. Hence, 
(1) there exists
at least one atom $x$ different from $c$ that appears in $A\cup C\cup D$ but not in $G'$, or
(2) there is some atom $y$ in $(B\cup E) \setminus F'$.
Let $\ell=x$ for case (1)
and $\ell=\neg y$ for case (2), and
consider $\Psi=\Phi\setminus \{\eneg \ell\}$.
Since $\Psi\subseteq \Phi$, $r^\Psi$ and $s^\Psi$ are non-empty.
We observe that
  $s^\Psi=A' \leftarrow B', \neg C', \neg D'^\Psi, \neg\neg E'^\Psi, 
\neg\neg F', \neg G'$ 
satisfies $c\notin D'^\Psi$ (for case (1) recall that $c\neq x$).
Hence, we have $c\in C$ but $c\notin C'\cup D'^\Psi \cup  G'$.
As before, 
condition ($\gamma$) 
of Lemma~\ref{lemma:subrule} is thus violated,
and we conclude
$\semods{\Pi_r^\Psi}\not\subseteq \semods{\Pi_s^\Psi}$. 
It remains to show that $\mods{\Pi_r^\Psi}$ realizes $\Psi$.
Here, any interpretation $I$ that satisfies $\ell$
(i.e.\ where $x \in I$ in case (1); otherwise $y \notin I$)
is a model 
of $r^\Psi$.  
We take $\calI$ as the set of interpretations 
in which 
$\ell$ is true,
all $f\in F'$ are set to true, and 
all $g\in G'$ are set to false. 
Indeed, $\calI$ contains well-formed interpretations, since 
$x\notin G'$ and likewise, $y\notin F'$.
As before, one can show that $\calI$ is $\Psi$-compatible.
%Indeed, $\calI\subseteq \mods{r^\Psi}$.

Suppose (b) is violated. 
We consider $\Psi=\Phi\setminus \{\eneg \neg b\}$ with 
$b\in B\setminus (B'\cup E'\cup F')$.  
Since $\Psi\subseteq \Phi$, $r^\Psi$ and $s^\Psi$ are non-empty.
In fact, $s^\Psi=s^\Phi$ since $b\notin E'$.
Thus, condition ($\beta$) of Lemma~\ref{lemma:subrule} is violated,
and, by the same lemma, 
we have $\semods{\Pi_r^\Psi}\not\subseteq \semods{\Pi_s^\Psi}$. 
It remains to show that $\mods{\Pi_r^\Psi}$ realizes $\Psi$.
Indeed, any interpretation that sets $b$ to false is a model 
of $r^\Psi$.  
We take $\calI$  as the set of interpretations 
where
$b$ is set to false,
all $f\in F'$ are set to true, and 
all $g\in G'$ are set to false. 
This is again sound since $b\notin F'$.
As before, one can show that $\calI$ is $\Psi$-compatible.

Suppose (b$^*$) is violated, i.e.\
we have 
$(A\cup C\cup D)\setminus G'\neq \emptyset$ or
$|(B\cup E)\setminus F'|> 1$,
but
$B\not\subseteq B'\cup F'$.
Let $b\in B\setminus(B'\cup F')$. Hence, there exists
(1) an atom $x$ in $(A\cup C\cup D) \setminus G'$, or
(2) an atom $y$ different to $b$ that appears in $(B\cup E)\setminus F'$.
Let $\ell=x$ in case (1)
and $\ell=\neg y$ in case (2), and
consider $\Psi=\Phi\setminus \{\eneg \ell\}$.
Since $\Psi\subseteq \Phi$, $r^\Psi$ and $s^\Psi$ are non-empty.
We observe that
  $s^\Psi=A' \leftarrow B', \neg C', \neg D'^\Psi, \neg\neg E'^\Psi, 
\neg\neg F', \neg G'$ 
satisfies $b\notin E'^\Psi$ (for case (2) recall that $b\neq y$).
Hence, we have $b\in B$ but $b\notin B'\cup E'^\Psi \cup  F'$.
As before, 
condition ($\beta$) 
of Lemma~\ref{lemma:subrule} is violated,
and we conclude
$\semods{\Pi_r^\Psi}\not\subseteq \semods{\Pi_s^\Psi}$. 
It remains to show that $\mods{\Pi_r^\Psi}$ realizes $\Psi$.
Here, any interpretation that sets $\ell$ to true 
(i.e.\ $x$ is set true in case (1); otherwise $y$ is set to false)
is a model 
of $r^\Psi$.  
We take $\calI$  as the set of all interpretations 
wherein
$\ell$ is true,
all $f\in F'$ are set to true, and 
all $g\in G'$ are set to false. 
Indeed, $\calI$ contains well-formed interpretations, since 
$x\notin G'$ and likewise, $y\notin F'$.
One can show that $\calI$ is $\Psi$-compatible by the usual argument.
%Indeed, $\calI\subseteq \mods{r^\Psi}$.

Suppose (b') is violated, i.e.\ 
we have 
$X=A \cap (A'\setminus (C'\cup D'\cup G'))\neq\emptyset$ and $B\not\subseteq B'$. 
Consider $\Psi=\Phi\setminus \{\eneg a\}$ 
with $a\in X$. In particular, this implies $a\notin D'$.
Since $\Psi\subseteq \Phi$, $r^\Psi$ and $s^\Psi$ are non-empty.
In fact, $s^\Psi=s^\Phi$ since $a\notin D'$.
We observe 
that condition ($\beta'$) of Lemma~\ref{lemma:subrule} is thus violated,
and, by the same lemma, 
we get that $\semods{\Pi_r^\Psi}\not\subseteq \semods{\Pi_s^\Psi}$. 
It remains to show that $\mods{\Pi_r^\Psi}$ realizes $\Psi$.
Indeed, any interpretation that sets $a$ to true is a model 
of $r^\Psi$.  
We take $\calI$  as the set of interpretations 
where
$a$ is set to true,
all $f\in F'$ are set to true, and 
all $g\in G'$ are set to false. 
This is again sound since $a\notin G'$ (by construction of $X$).

Suppose (b$^*$') is violated, i.e.\ 
$|(A\cup C\cup D)\setminus G'|>1$ or
$(B\cup E)\setminus F'\neq \emptyset$,
as well as 
$X=A \cap (A'\setminus (C'\cup G'))\neq\emptyset$ and $B\not\subseteq B'$. 
For this case, we consider the rules $r^\Phi$ and $s^\Phi$ directly.
In fact, for these rules, $X$ and $B\not\subseteq B'$ 
immediately imply that condition ($\beta'$) of
Lemma~\ref{lemma:subrule} is violated,
and we conclude
$\semods{\Pi_r^\Phi}\not\subseteq \semods{\Pi_s^\Phi}$. 
We show that $\calM = \mods{\Pi_r^\Phi}$ realizes $\Phi$.
Note that in this particular case, $r^\Phi$ has at least two literals not
``absorbed'' by $F'$ and $G'$; viz.\
$x\in X$ and some $\ell\neq x$ 
from 
$(A\cup C\cup D)\setminus (G'\cup \{x\})$ or
from
$(B\cup E)\setminus F'$.
If we consider the set $\calM' \subseteq \calM$ of models 
where all $f\in F'$ are set to true, and 
all $g\in G'$ are set to false, we observe that for all other atoms $y$, 
there exists at least one model in $\calM'$ where $y$ is 
set to false 
and another one where $y$ is set to true.
In particular, there is an $M\in\calM$ where $x$ is false (namely the one where $\ell$ is set to true),
and likewise there is an $M\in\calM$ where $\ell$ is false (namely the one where $x$ is set to true).
%Here, any intepretation that sets $\ell$ to true 
%(i.e.\ $x$ is set true in case (1); otherwise $y$ is set to false)
%is a model 
%of $r^\Psi$.  
%We take $\calI$  as the set of interpretations 
%where 
%$\ell$ is set to true,
%all $f\in F'$ are set to true, and 
%all $g\in G'$ are set to false. 
%Indeed, $\calI$ contains well-formed interpretations, since 
%$x\notin G'$ and likewiese, $y\notin F'$.
%One can show that $\calI$ is $\Psi$-compatible by the usual argument.
%Indeed, $\calI\subseteq \mods{r^\Psi}$.

Cases (a) and (a$^*$) are similar to (c) and (c$^*$): simply apply
Lemma~\ref{lemma:subrule} via condition ($\alpha$) instead of ($\gamma$).
%Suppose (a) is violated and 
%consider $r^\Phi$, $s^\Phi$. 
%$a\in A\setminus (A'\cup C' \cup G')$ appears in 
%$r^\Phi$ but does not appear in the head or negative body of 
%$s^\Phi$, i.e.\ 
%condition (a) of Lemma~\ref{lemma:subrule} is violated. Hence
%$\semods{r^\Phi}\not\subseteq \semods{s^\Phi}$. 
%It can be showm that 
%$\calI$, the set of interpretation where 
%$a$ is set to true,
%all $f\in F'$ are set to true, and 
%all $g\in G'$ are set to false, 
%satisfies
%$\calI\subseteq\models{r^Y}$ and
%is $\Psi$-compatible.

\clearpage
\bibliographystyle{aaai}
\bibliography{references}

\end{document}